\documentclass{sig-alternate-2013}
\makeatletter
\def\ps@headings{%
\def\@oddhead{\mbox{}\scriptsize\rightmark \hfil \thepage}%
\def\@evenhead{\scriptsize\thepage \hfil \leftmark\mbox{}}%
\def\@oddfoot{}%
\def\@evenfoot{}}
\makeatother
\usepackage{algpseudocode}
\usepackage{algorithm}
\usepackage{amsmath} 
\usepackage{graphicx}
\usepackage{subfigure}
\usepackage{bbm}
\usepackage{bm}
\usepackage{overpic}
\usepackage[hidelinks]{hyperref}

\newtheorem{definition}{Definition}
\newtheorem{lemma}{Lemma}
\newtheorem{theorem}{Theorem}
\newtheorem{corollary}{Corollary}

\newcommand{\capp}{\operatorname{cap}}
\newcommand{\lex}{\operatorname{lex}}

\makeatletter
\def\@copyrightspace{\relax}
\makeatother

\begin{document}

\title{Loop-Free Backpressure Routing \\ Using Link-Reversal Algorithms}
\numberofauthors{4}
\author{
\alignauthor
Anurag Rai\\
       \affaddr{LIDS, MIT}\\
       \affaddr{Cambridge, MA, USA}\\
       \email{rai@mit.edu}
\alignauthor
Chih-ping Li\\
	\affaddr{Qualcomm Research}\\
	\affaddr{San Diego, CA, USA}\\
	\email{cpli@qti.qualcomm.com}
\and
\alignauthor
Georgios Paschos\\
	\affaddr{Mathematical and Algorithmic Sciences Lab}\\
	\affaddr {France Research Center}\\
	\affaddr{Huawei Technologies Co., Ltd.}\\
	\email{georgios.paschos@huawei.com}
\alignauthor
Eytan Modiano\\
       \affaddr{LIDS,MIT}\\
       \affaddr{Cambridge, MA, USA}\\
       \email{modiano@mit.edu}
}

\maketitle
\begin{abstract}
The backpressure routing policy is known to be a throughput optimal policy that supports any feasible traffic demand in data networks, but may have poor delay performance when packets traverse loops in the network. In this paper, we study loop-free backpressure routing policies that forward packets along directed acyclic graphs (DAGs) to avoid the looping problem. These policies use link reversal algorithms to improve the DAGs in order to support any achievable traffic demand.

For a network with a single commodity, we show that a DAG that supports a given traffic demand can be found after a finite number of iterations of the link-reversal process. We use this to develop a joint link-reversal and backpressure routing policy, called the loop free backpressure (LFBP) algorithm. This algorithm forwards packets on the DAG, while the DAG is dynamically updated based on the growth of the queue backlogs. We show by simulations that such a DAG-based policy improves the delay over the classical backpressure routing policy. We also propose a multicommodity version of the LFBP algorithm, and via simulation we show that its delay performance is better than that of backpressure.

\end{abstract}
\section{Introduction}
Throughput and delay are the two major  metrics used to evaluate the performance of communication networks.
For networks that exhibit high variability, such as mobile ad~hoc networks, the dynamic backpressure routing policy  \cite{tassiulas} 
is a highly desirable solution,  known to maximize  throughput in a wide range of settings. 
However,  the delay performance of backpressure is  poor \cite{srikant}.
 The high delay is attributed to  a property of  backpressure that allows the packets to loop within the network instead of moving towards the destination. 
 In this paper we  improve the delay performance of backpressure routing by constraining the data routing along loop free paths.

To eliminate  loops in the network, we assign directions to the links such that the network becomes a directed acyclic graph (DAG). 
Initially, we generate an  arbitrary DAG and use backpressure routing over it.
If the initial DAG has max-flow smaller than  the traffic demand, parts of the network become overloaded. 
By reversing the direction of the links that point from non-overloaded to overloaded nodes a new DAG with a lower overload is obtained. 
Iterating over this process, our distributed algorithm gradually converges to a DAG that supports any  traffic demand  feasible in the network. Hence the loop-free property is achieved without the loss of throughput.

Prior work  identifies  looping as a main cause for high delays in backpressure routing and proposes delay-aware backpressure techniques. 
Backpressure enhanced with hop count bias is first proposed  in \cite{neely_enhanced} to drive packets through  paths with smallest hop counts when the load is low.  
An alternative backpressure modification that utilizes shortest path information is proposed in \cite{ying}. 
A different line of works proposes to learn the network topology using backpressure and then use this information to enhance routing decisions. In \cite{atilla} backpressure is constrained to a subgraph which is discovered by running unconstrained backpressure for a time period and computing the average number of  packets routed over each link. Learning is  effectively used in scheduling \cite{huang} and utility optimization \cite{longbo} for wireless networks.
In our work we aim to eliminate loops by restricting backpressure to a DAG, while we dynamically improve the DAG by reversing links.

The link-reversal algorithms were introduced in \cite{gafni_bertsekas} as a means to maintain connectivity in networks with volatile links.  These distributed algorithms react to any topological changes to obtain a DAG such that each node has a loop-free path  to the destination.  In \cite{tora}, one of the link-reversal algorithms was used to design a routing protocol (called TORA) for multihop wireless networks. Although these algorithms provide loop free paths and guarantee connectivity from the nodes to the destination, they do not maximize throughput. Thus, the main goal of this paper is to create a new link-reversal algorithm and combine it with the backpressure algorithm to construct a distributed throughput optimal algorithm with improved delay performance.

The main contributions of this paper are as follows:
\begin{itemize}
\itemsep0em 
\item  For a DAG with a single commodity, we study the  lexicographic optimization of the queue growth rate.
We develop a novel link-reversal algorithm that reverses link direction based on overload conditions to form a new DAG
with lexicographically smaller queue growth rates. 
\item We show that the queue backlog information from backpressure routing can be used to perform the link reversals. We develop loop free backpressure (LFBP) algorithm, a  distributed routing scheme  that eliminates loops and  retains the throughput optimality property.
\item Our simulation results of LFBP show a significant delay improvement over backpressure in static and dynamic networks. 
\item We extend the LFBP algorithm to networks with multiple commodities, and provide a simulation result to show its delay improvement over backpressure.
 \end{itemize}

\section{System Model and Definitions}
\subsection{Network model}
We consider the problem of routing single-commodity data packets in a network. The network is represented by a graph $G=(N,E)$, where $N$ is the set of nodes and $E$ is the set of undirected links $\{i, j\}$ with capacity $c_{ij}$. Packets arrive at the source node $s$ at rate $\lambda$  and are destined for a receiver node $d$.  Let $f^{\max}$ denote the maximum flow from node $s$ to  $d$ in the network $G$.
The quantity $f^{\max}$ is the maximally achievable throughput at the destination node $d$. 

To avoid unnecessary routing loops, we restrict forwarding along a directed acyclic graph (DAG) embedded in the graph $G$. 
An optimal DAG exists to support the max-flow $f^{\max}$ and can be found by: 
(i) computing a feasible flow allocation $(f_{ij})$ that yields the max-flow $f^{\max}$ in  $G$ (e.g.~using \cite{FordFulkerson}); 
(ii) trimming any positive flow on  directed cycles; 
(iii) defining an embedded DAG by assigning  a direction for each link $\{i, j\}$ according to the direction of the flow $f_{ij}$ on that link. 
Since backpressure achieves the max-flow of a constrained graph \cite{neely_georgiadis}, performing backpressure routing over the optimal DAG supports $\lambda$.


This centralized approach is unsuitable for mobile ad-hoc networks, which are based on wireless links with  time-varying  capacities and may undergo frequent topology changes. 
In such situations,  the  optimal embedded DAG  also changes with time, which  requires  constantly repeating the above offline process.
 Instead, 
it is possible to use a distributed  adaptive mechanism that   reverses the direction of links until  a DAG that supports the current traffic demand is found.
In this paper we propose  an  algorithm that  reacts to the traffic conditions by changing the directions of some links. 
 To understand the properties of 
 the link-reversing operations, we first study the fluid level behavior of a network under overload conditions.

\subsection{Flow equations}

Consider an embedded DAG $D_{k} = (N_{k}, E_{k})$ in the network graph $G$, where $N_{k} = N$ is the set of network nodes and $E_{k}$ is the set of directed links.\footnote{The notation $D_{k}$ of an embedded DAG is useful in the paper; it will denote the DAG that is formed after the $k$th iteration of the link-reversal algorithm.} For each link $\{i, j\}\in E$, either $(i, j)$ or $(j, i)$ belongs to $E_{k}$ (but not both). Each directed link $(i, j)$ has the capacity of the undirected counterpart $\{i,j\}$, which is $c_{ij}$.
Let $f^{\max}_{k}$ be the maximum flow of the DAG $D_{k}$ from the source node $s$ to the destination node $d$.
Any embedded DAG has smaller or equal max-flow with respect to $G$,  $f^{\max}_{k}\leq f^{\max}.$ 

For two disjoint subsets $A$ and $B$ of nodes in $D_{k}$, we define $\capp_k(A, B)$ as the total capacity of the directed links going from $A$ to $B$, i.e.,
\begin{equation} \label{eq:202}
\capp_k(A, B) = \sum_{(i, j) \in E_{k}: i\in A,\, j\in B} c_{ij}.
\end{equation}

A cut 
 is 
 a partition of nodes $(A,A^c)$ such that $s \in A$ and $d \in A^c$. A cut $(A_k,A_k^{c})$ is 
 a min-cut if 
it minimizes the expression $\capp_k(A_k,A_k^{c}) $ over all cuts. 
By the max-flow min-cut theorem $f^{\max}_k = \capp_k(A_k,A_k^{c})$, where $(A_k,A_k^{c})$ is the min-cut of the DAG $D_k$.
We remark that a cut in a DAG is also a cut in $G$ or another DAG.
However, the value of $\capp_k(.,.)$ depends on the graph considered (see summation in (\ref{eq:202})).

We consider the network as a time-slotted system, where slot $t$ refers to the time interval $[t, t+1)$, $t\in\{0, 1, 2, \ldots\}$. Each network node $n$ maintains a queue $Q_{n}(t)$, where $Q_{n}(t)$ also denotes the queue backlog at time $t$. We have $Q_{d}(t)=0$ for all $t$ since the destination node $d$ does not buffer packets. Let $A(t)$ be the number of exogenous packets arriving at the source node $s$ in slot $t$. Under a routing policy that forwards packets over the directed links defined by the DAG $D_{k}$, let $F_{ij}(t)$ be the number of packets that are transmitted over the directed link $(i, j) \in E_{k}$ in slot $t$; the link capacity constraint states that $F_{ij}(t)\leq c_{ij}$ for all $t$. The queues $Q_{n}(t), n\ne d,$ are updated over slots according to
\begin{multline} \label{eq:101}
Q_n(t) = Q_{n}(t-1) + 1_{[n=s]} A(t) \\
+ \sum_{i:(i,n) \in E_{k} } F_{in}(t) - \sum_{j:(n,j) \in E_{k} } F_{nj}(t),
\end{multline}
where $1_{[\cdot]}$ is an indicator function. 

To study the overload behavior of the system we define
the queue overload (i.e., growth) rate at node $n$ as
\begin{equation} \label{eq:104}
q_n = \lim_{t\rightarrow \infty} \frac{Q_n(t)}{t}.
\end{equation}
Additionally, define the exogenous packet arrival rate $\lambda$ and  the flow $f_{ij}$ over a directed link $(i, j)$  as
\begin{equation*} \label{eq:102}
\lambda = \lim_{t\rightarrow \infty} \frac{1}{t}\sum_{\tau=0}^{t-1} A(\tau),\quad f_{ij} = \lim_{t\rightarrow \infty} \frac{1}{t}\sum_{\tau=0}^{t-1} F_{ij}(\tau),
\end{equation*}
where the above limits are assumed to exist almost surely (see \cite{georgiadis} for details). Using the recursion (\ref{eq:101}), taking time averages and letting $t\to\infty$, we have the fluid-level equation:
\begin{align} 
q_n &=  1_{[n=s]} \lambda+ \sum_{i: (i,n) \in E_k}  f_{in} - \sum_{j: (n,j) \in E_k}  f_{nj},\ \forall n \in N\backslash \{d\} \label{eq:105} \\ 
0 &\le f_{ij} \le c_{ij},\ \forall (i, j) \in E_k. \label{eq:106}
\end{align}
Equations~\eqref{eq:105} and~\eqref{eq:106} are the flow conservation and link capacity constraints, respectively. A network node $n$ is said to be overloaded if its queue growth rate $q_{n}$ is positive, which implies that $Q_n(t) \rightarrow \infty$ as $t \rightarrow \infty$ (see~\eqref{eq:104} and~\cite{Nee10book}). Summing~\eqref{eq:105} over $n\in N$ yields
\begin{equation} \label{eq:107}
\sum_{n\in N\backslash \{d\}} q_{n} = \lambda - \sum_{i: (i, d)\in E_{k}} f_{id},
\end{equation}
where $\sum_{i: (i, d)\in E_{k}} f_{id}$ denotes the throughput received at the destination  $d$.
Therefore, equation~\eqref{eq:107} states that the received throughput 
is equal to the exogenous arrival rate $\lambda$ less the sum of queue growth rates $\sum_{n\in N} q_{n}$ in the network.

\subsection{Properties of queue overload vector}

If the traffic arrival rate $\lambda$ is strictly larger than the maximum flow $f^{\max}_{k}$ of the DAG $D_{k}$, then some network nodes will be overloaded. It is because, from~\eqref{eq:107}, we have
\begin{equation} \label{eq:201}
\sum_{n\in N} q_{n} = \lambda - \sum_{i: (i, d) \in E_{k}} f_{id} \geq \lambda - f^{\max}_{k} > 0,
\end{equation}
which implies that $q_{n}>0$ for some node $n\in N$. Let $\bm{q}=(q_{n})_{n\in N}$ be the queue overload vector. A queue overload vector $\bm{q}$ is feasible in the DAG $D_{k}$ if there exist overload rates $(q_{n})_{n\in N}$ and flow variables $(f_{ij})_{(i, j)\in E_{k}}$ that satisfy~\eqref{eq:105} and~\eqref{eq:106}. Let $\mathcal{Q}_{k}$ be the set of all feasible queue overload vectors in $D_{k}$. We are interested in the \emph{lexicographically smallest} queue overflow vector in set $\mathcal{Q}_{k}$. Formally, given a vector $\bm{u} = (u_{1}, \ldots, u_{N})$, let $\overline{u}_{i}$ be the $i$th maximal component of $\bm{u}$. We say that a vector $\bm{u}$ is \emph{lexicographically smaller} than a vector $\bm{v}$, denoted by $\bm{u} <_{\lex} \bm{v}$, if $\overline{u}_{1} < \overline{v}_{1}$ or $\overline{u}_{i} = \overline{v}_{i}$ for all $i=1, ..., (j-1)$ and $\overline{u}_{j} < \overline{v}_{j}$ for some $j=2,\dots, N$. If $\overline{u}_{i} = \overline{v}_{i}$ for all $i$, then the two vectors are lexicographically equal, represented by $\bm{u} =_{\lex} \bm{v}$.\footnote{As an example, the two vectors $\bm{u}=(3, 2, 1, 2, 1)$ and $\bm{v}=(1, 2, 3, 2, 2)$ satisfy $\bm{u} <_{\lex} \bm{v}$ because $\overline{u}_{1} = \overline{v}_{1} = 3$, $\overline{u}_{2} = \overline{v}_{2} = \overline{u}_{3} = \overline{v}_{3} = 2$, and $\overline{u}_{4} =1 < \overline{v}_{4} = 2$.} 
The above-defined vector comparison induces a total order on the set  $\mathcal{Q}_{k}$, and hence the existence of a lexicographically smallest vector is always   guaranteed \cite{georgatsos}.

\begin{lemma}[\cite{georgiadis}] \label{lem:201}
Let $\bm{q}_{k}^{\text{min}}$ be the lexicographically smallest vector in the queue overload region $\mathcal{Q}_{k}$ of the DAG $D_{k}$. We have the following properties:
\begin{enumerate}
\itemsep0em 
\item The vector $\bm{q}_{k}^{\text{min}}$ exists and is unique in the set $\mathcal{Q}_{k}$.
\item The vector $\bm{q}_{k}^{\text{min}}$ minimizes the sum of queue overload rates, i.e., it is a solution to the optimization problem:
\[
\text{minimize } \sum_{n\in N} q_{n}, \ \text{subject to } \bm{q} \in \mathcal{Q}_{k}
\]
(direct consequence of Theorem 1 in \cite{georgiadis}). Due to~\eqref{eq:107}, the corresponding throughput is maximized.
\item  A feasible flow allocation vector $(f_{ij})_{(i, j) \in E_{k}}$ induces $\bm{q}_{k}^{\text{min}}$ if and only if over each link $(i, j)\in E_{k}$ the following holds:
\begin{align}
&\text{if } q_i < q_j, \text{ then } f_{ij} = 0, \label{property1}\\
&\text{if }q_i > q_j, \text{ then } f_{ij} = c_{ij} \label{property2}.
\end{align} 
\end{enumerate}

\end{lemma}

In general, there are many flow allocations that yield the maximum throughput. Focusing on those that additionally induce $\bm{q}^{\text{min}}_{k}$ has two advantages. First, as  shown next, these allocations lead to link-reversal operations that improve the max-flow of the DAG $D_{k}$. Second, the backpressure algorithm can be used to preform the same reversals and improve the max-flow; we will use this observation in Section~\ref{dynamic algorithm} to combine link-reversal algorithms with backpressure routing.

\section{Link-Reversal Algorithms} \label{sec:link_reversal}

The link-reversal algorithms given in \cite{gafni_bertsekas} were designed to maintain a path from each node in the network to the destination. One algorithm relevant to this paper is  the \emph{full reversal method}. This algorithm is triggered when some nodes $n \ne d$ lose all of their outgoing links. At every iteration of the algorithm, nodes $n$, that have no outgoing link, reverse the direction of all their incoming links. This process is repeated until all the nodes other than the destination have at least one outgoing link. When the process stops these nodes are guaranteed to have a path to the destination.
The example in Figure \ref{fig:gb_example}, taken from \cite{gafni_bertsekas}, illustrates this algorithm at work. 
\begin{figure}[h!]
\centering
\subfigure[]{
\begin{overpic}[scale=.42]	{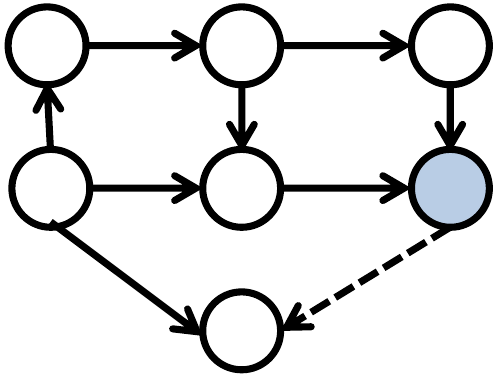}
	\put (45,4) {\small $d$}
\end{overpic}
\label{fig:gb0}
}
\subfigure[]{
\begin{overpic}[scale=.42]	{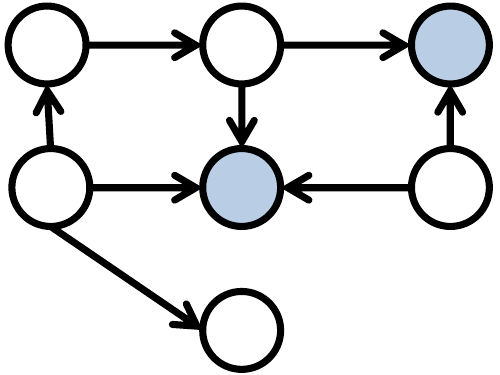}
	\put (45,4) {\small $d$}
\end{overpic}
\label{fig:gb1}
}
\subfigure[]{
\begin{overpic}[scale=.42]	{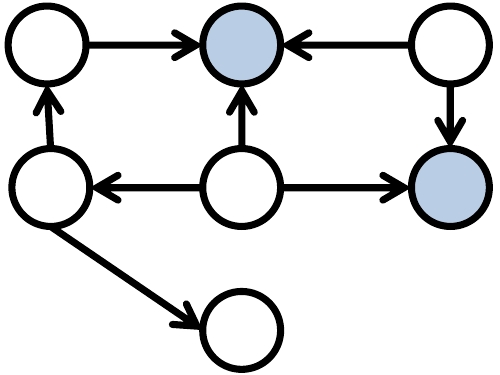}
	\put (45,4) {\small $d$}
\end{overpic}
\label{fig:gb2}
}
\subfigure[]{
\begin{overpic}[scale=.42]	{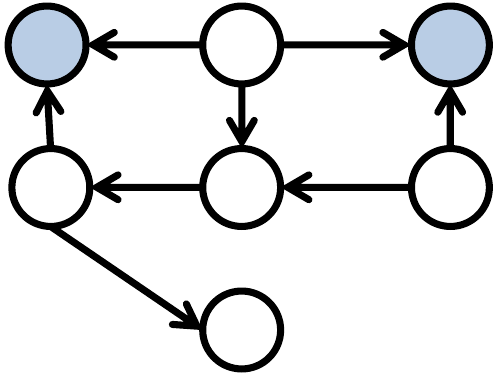}
	\put (45,4) {\small $d$}
\end{overpic}
\label{fig:gb3}
}
\subfigure[]{
\begin{overpic}[scale=.42]	{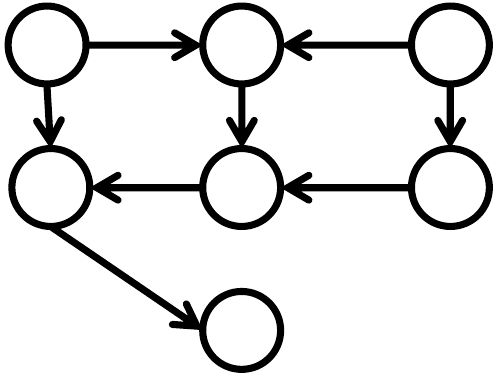}
	\put (45,4) {\small $d$}
\end{overpic}
\label{fig:gb4}
}
\caption{Illustration of the full reversal method of \cite{gafni_bertsekas} when the dashed link in Figure \ref{fig:gb0}  is lost. At every iteration, the algorithm reverses all the links incident to the nodes with no outgoing link.}
\label{fig:gb_example}
\end{figure}

Although the full reversal algorithm guarantees connectivity, the resulting throughput may be significantly lower than the maximum possible. Hence, in this paper we shift the focus from connectivity to maximum throughput. Specifically, we propose a novel link-reversal algorithm that produces a DAG which supports the traffic demand $\lambda$, assuming $\lambda \le f^{\max}$. We do this by quickly constructing an initial DAG and improving upon it in multiple iterations.

\subsection{Initial DAG}  \label{sec: initial dag}
We assume that each node in the network has a unique ID. These IDs give a topological ordering to the nodes. So, the initial DAG can be created simply by directing each link to go from the node with the lower ID to the node with the higher ID. If the unique IDs are not available, the initial DAG can be created by using a strategy such as the one given in \cite{tora}.

\subsection{Overload detection}

Given a DAG $D_{k}$, $k=0, 1, 2, \ldots$, we suppose that there is a routing policy $\pi$ that yields the lexicographically minimal queue overload vector $\bm{q}^{\text{min}}_{k}$. \footnote{Such a policy $\pi$ can simply solve an optimization problem offline to compute the required flow allocation. In Section~\ref{dynamic algorithm}, we develop a distributed algorithm using backpressure that does not require the computation of the lexicographically optimal overload vector. We use this vector only to prove the properties of our link-reversal algorithm.} 
Then we use the vector $\bm{q}^{\text{min}}_{k}$
to detect node overload and decide whether a link should be reversed.

If the data arrival rate $\lambda$ is less than or equal to the maximum flow $f^{\max}_{k}$ of the DAG $D_{k}$, then there exists a flow allocation $(f_{ij})$ that supports the traffic demand and yields zero queue overload rates $q_{n}=0$ at all nodes $n\in N$. By  the second property of Lemma~\ref{lem:201} and nonnegativity of the overload vector,
 the queue overload vector $\bm{q}^{\text{min}}_{k}$ is zero.
 Thus, the throughput under policy $\pi$ is  $\lambda$  according to~\eqref{eq:107}, and the current DAG $D_{k}$  supports $\lambda$; no link-reversal operations are needed. 

On the other hand, if the arrival rate $\lambda$ is strictly larger than the maximum flow $f^{\max}_{k}$, 
by the second property in Lemma~\ref{lem:201}  
the maximum throughput is $f^{\max}_{k}$ and
the queue overload vector $\bm{q}^{\min}_{k} = (q_{k,n}^{\text{min}})_{n\in N}$ is nonzero because we have  from~\eqref{eq:201} that
\[
\sum_{n\in N} q_{k,n}^{\text{min}} > \lambda - f^{\max}_{k} > 0.
\]
We may therefore detect the event ``DAG $D_k$  supports $\lambda$'' by testing whether  the overload vector 
$\bm{q}^{\min}_{k}$ is zero or  non-zero.

The next lemma shows that if DAG $D_k$  does not support $\lambda$ then it contains at least one  under-utilized link (our link-reversal algorithm will reverse the direction of such links to improve network throughput).
\begin{lemma}\label{lem: backward link}
Suppose that the traffic demand $\lambda$ satisfies
\[
f^{\max}_{k}<\lambda\leq f^{\max}.
\]
where $f^{\max}_{k}$ is the max-flow of the DAG $D_{k}$ and $f^{\max}$ is the max-flow of the undirected network $G$.
Then 
 there exists a link $(i, j)\in E_{k}$ such that $q_{k,i}^{\text{min}}=0$ and $q_{k,j}^{\text{min}}>0$.
\end{lemma}

\begin{proof}[of Lemma~\ref{lem: backward link}]
Let $A_{k}$ be the set of overloaded nodes  under a flow allocation that induces the lexicographically minimal overload vector ${\bm q}^{\min}_{k}$ in the DAG $D_{k}$; the set $A_{k}$ is nonempty due to $\lambda > f^{\max}_{k}$ and~\eqref{eq:201}. It follows that the partition $(A_{k}, A_{k}^{c})$ is a min-cut of $D_{k}$ (see Lemma~\ref{MIN-CUT} in the Appendix).\footnote{The set $A_{k}^{c}$ contains the destination node $d$ and is nonempty.} By the max-flow min-cut theorem, the capacity of the min-cut $(A_{k}, A_{k}^{c})$ in $D_{k}$ satisfies $\capp_k(A_{k}, A_{k}^c) =  f^{\max}_k < f^{\max}$.

The proof is by contradiction. Let us assume that there is no directed link that goes from the set $A_{k}^c$ to $A_{k}$ in the DAG $D_k$. It follows that $\capp_k(A_{k}, A_{k}^c)$ is the sum of capacities of all undirected links between the sets $A_{k}$ and $A_{k}^c$, i.e., 
\[
\capp_k(A_{k}, A_{k}^{c}) = \sum_{i\in A_{k},\, j\notin A_{k}} c_{ij},
\]
which is equal to the value of the cut $(A_k,A_k^c)$ in graph $G$.
Since the value of any cut is larger or equal to the min-cut, applying the max-flow min-cut theorem on $G$ we have
\begin{align*}
f^{\max} \leq  \sum_{i\in A_{k},\, j\notin A_{k}} c_{ij} =
 \capp_k(A_{k}, A_{k}^{c}) = f_{k}^{\text{max}},
\end{align*}
which contradicts the assumption that $f_{k}^{\max} < \lambda \leq f^{\max}$.
\end{proof}

\subsection{Link reversal}

Lemma~\ref{lem: backward link} shows that if the DAG $D_{k}$ has insufficient capacity to support the traffic demand $\lambda \leq f^{\max}$, then there exists a directed link from an underloaded node $i$ to an overloaded one $j$ under the lexicographically minimum overflow vector $\bm{q}_{k}^{\text{min}}$. 
Because of property \eqref{property1}, we may infer that this link is not utilized. Next we show that reversing the direction of this link provides a strictly improved DAG.

We consider the link-reversal algorithm (Algorithm \ref{alg: link reversal}) that reverses all such links that satisfy the property in Lemma~\ref{lem: backward link}. This reversal yields a new directed graph $D_{k+1} = (N, E_{k+1})$.
\begin{algorithm}[h!] 
\caption{Link-Reversal Algorithm}
\label{alg: link reversal}
\begin{algorithmic}[1]
	\ForAll {$(i,j) \in E_k$}
		\If{$q_{k, i}^{\min}=0 \text{ and } q_{k, j}^{\min} > 0$} 
			\State $(j,i) \in E_{k+1} $
		\Else
			\State $(i,j) \in E_{k+1}$
		\EndIf
	\EndFor
\end{algorithmic}
\end{algorithm}

\begin{lemma} \label{lem: graph is dag}
The directed graph $D_{k+1}$ is acyclic.
\end{lemma}
\begin{proof}[of Lemma~\ref{lem: graph is dag}]
Recall that  $A_{k}$ is the set of overloaded nodes in the DAG $D_{k}$ under the lexicographically minimum queue overload vector $\bm{q}_{k}^{\text{min}}$. Let $L_{k}\subseteq E$ be the set of undirected links between $A_{k}$ and $A_{k}^{c}$. Algorithm~\ref{alg: link reversal} 
changes the link direction in a subset of $L_{k}$. More precisely, it 
enforces the direction of all links in $L_{k}$ to go from $A_{k}$ to $A_{k}^{c}$.

We complete the proof by construction in two steps. First, we remove all links in $L_{k}$ from the DAG $D_{k}$, resulting in two disconnected subgraphs that are DAGs themselves. Second, consider that we add a link in $L_{k}$ back to the network with the direction going from $A_{k}$ to $A_{k}^{c}$. This link addition does not create a cycle because there is no path from $A_{k}^{c}$ to $A_{k}$, and the resulting graph remains to be a DAG. We can add the other links in $L_{k}$ one-by-one back to the graph with the direction from $A_{k}$ to $A_{k}^{c}$; similarly, these link additions do not create cycles. The final directed graph is $D_{k+1}$, and it is a DAG. See Fig.~\ref{fig:reversal} for an illustration.
\end{proof}

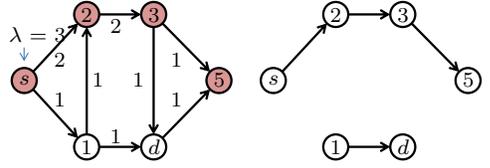
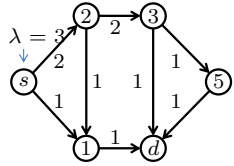
\begin{figure}[h]
\centering
\subfigure[The DAG $D_k$ with $A_{k}=\{s, 2, 3, 5\}$.]{
\begin{overpic}[scale=.42]	{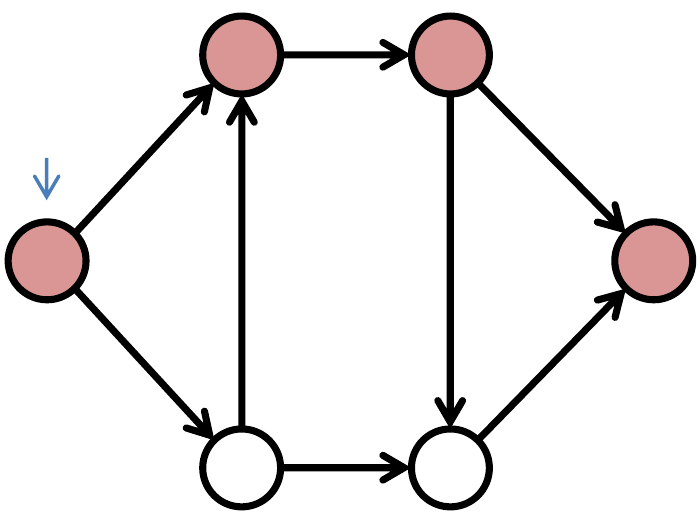}
	\put (4.5,36) {\small $s$}
	\put (32,64) {\small 2}
	\put (62,64) {\small 3}
	\put (91,35) {\small 5}
	\put (32,5) {\small 1}
	\put (62,5) {\small $d$}

	\put (0,55) {\small $\lambda$ = 3}
	\put (20,44) {\small 2}
	\put (20,26) {\small 1}

	\put (72,44) {\small 1}
	\put (72,26) {\small 1}

	\put (37,35) {\small 1}
	\put (55,35) {\small 1}

	\put (45,59) {\small 2}
	\put (45,10) {\small 1}
\end{overpic}
\label{fig:before}
}
\subfigure[Two disconnected DAGs formed by removing all links between $A_{k}$ and $A_{k}^{c}$.]{
\begin{overpic}[scale=.42]	{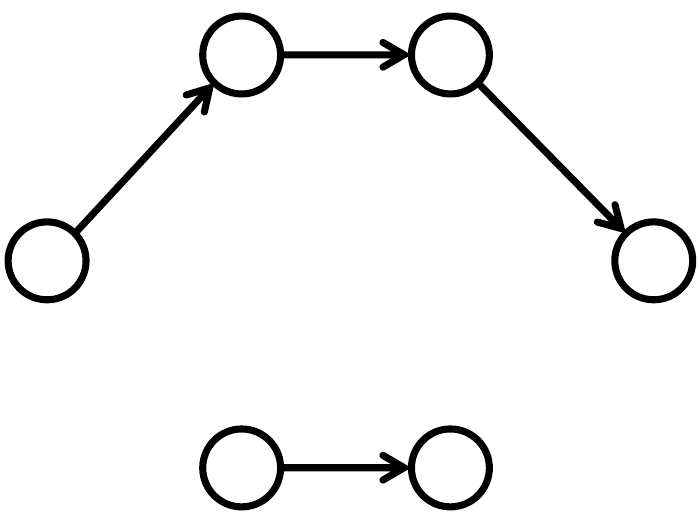}
	\put (4.5,36) {\small $s$}
	\put (32,64) {\small 2}
	\put (62,64) {\small 3}
	\put (91,35) {\small 5}
	\put (32,5) {\small 1}
	\put (62,5) {\small $d$}
\end{overpic}
\label{fig:separated}
}
\subfigure[The DAG $D_{k+1}$ formed by adding all links in $L_{k}$ back to the graph with the direction going from $A_{k}$ to $A_{k}^{c}$.]{
\begin{overpic}[scale=.42]	{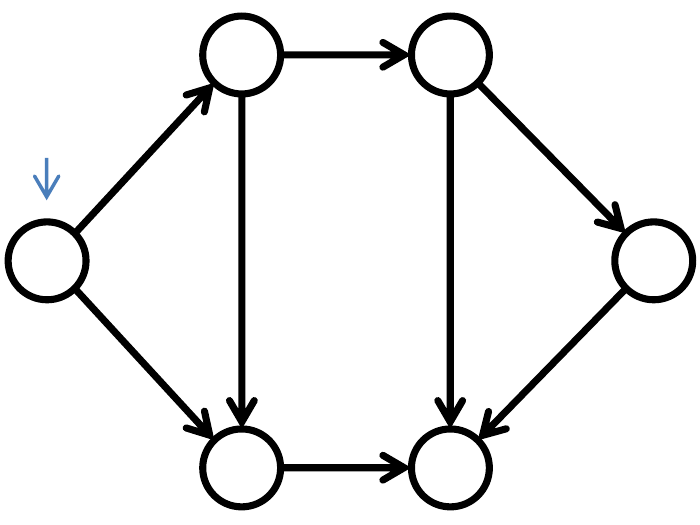}
	\put (4.5,36) {\small $s$}
	\put (32,64) {\small 2}
	\put (62,64) {\small 3}
	\put (91,35) {\small 5}
	\put (32,5) {\small 1}
	\put (62,5) {\small $d$}

	\put (0,55) {\small $\lambda$ = 3}
	\put (20,44) {\small 2}
	\put (20,26) {\small 1}

	\put (72,44) {\small 1}
	\put (72,26) {\small 1}

	\put (37,35) {\small 1}
	\put (55,35) {\small 1}

	\put (45,59) {\small 2}
	\put (45,10) {\small 1}
\end{overpic}
\label{fig:after}
}
\caption{Illustration for the proof of Lemma~\ref{lem: graph is dag}.}
\label{fig:reversal}
\end{figure}

The next lemma shows that the new DAG $D_{k+1}$ supports a lexicographically \emph{smaller} optimal overload vector (and therefore potentially better throughput) than the DAG $D_{k}$.
\begin{lemma} \label{lem: monotonicity}
Let $D_k$ be a DAG with the maximum flow $f_{k}^{\max} < \lambda \le f^{\max}$. The DAG $D_{k+1}$, obtained by performing  Algorithm~\ref{alg: link reversal} over $D_{k}$, has the lexicographically minimum queue overload vector satisfying $\bm{q}_{k+1}^{\text{min}} <_{\lex} \bm{q}_{k}^{\text{min}}$.
\end{lemma}
\begin{figure}[h]
\centering
\subfigure[Link $(a,b)$ before the link reversal.]{
\centering
\begin{overpic}[scale=.45, trim=0in 0in 0in .1in]	{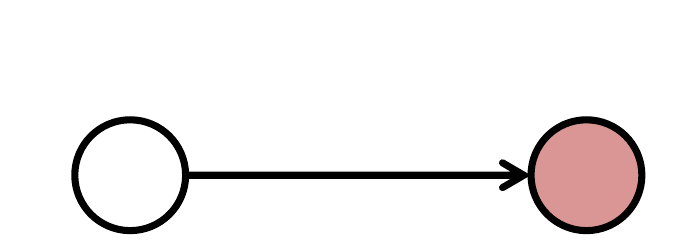}
	\put (16,8) {\small $a$}
	\put (81,8) {\small $b$}

	\put (5,25) {\small $q_a^{\min} = 0$}
	\put (70,25) {\small $q_b^{\min} > 0$}

	\put (40,15) {\small $f_{ab} = 0$}
\end{overpic}
}
\subfigure[Link $(b,a)$ after the link reversal.]{
\centering

\begin{overpic}[scale=.45, trim=0in 0in 0in .1in]{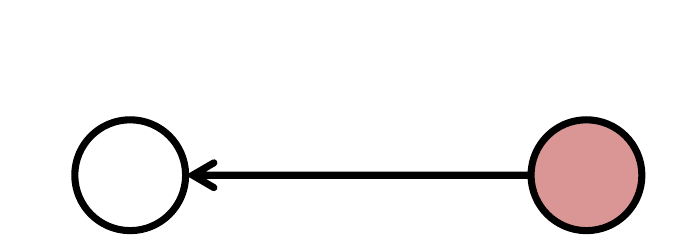}
	\put (16,8) {\small $a$}
	\put (81,8) {\small $b$}

	\put (5,25) {\small $\widehat{q}_a = \epsilon$}
	\put (65,25) {\small $\widehat{q}_b = q_b^{\min}-\epsilon$}

	\put (40,15) {\small $f'_{ba} = \epsilon$}
\end{overpic}
}

\caption{A link $\{a, b\}$ in the network in Fig.~\ref{fig:reversal} before and after link reversal. Before the reversal, the flow $f_{ab}$ is zero on $(a, b)$.  After the reversal, an $\epsilon$ flow can be sent over $(b, a)$ so that $(\widehat{q}_a, \widehat{q}_b) <_{\lex} (q_{k,a}^{\min}, q_{k, b}^{\min})$, while the rest of the flow allocation remains the same.}
\label{fig: monotonicity}
\end{figure}
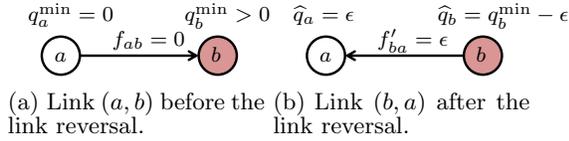
\begin{proof}[of Lemma~\ref{lem: monotonicity}]
Consider a link $(a, b) \in E_{k}$ such that $q_{k,a}^{\text{min}}=0$ and  $q_{k,b}^{\text{min}}>0$; this link exists by Lemma~\ref{lem: backward link}. From the property (\ref{property1}), any feasible flow allocation $(f_{ij})$ that yields the lexicographically minimum overload vector $\bm{q}_{k}^{\text{min}}$ must have $f_{ab}=0$ over link $(a, b)$. The link-reversal algorithm reverses the link $(a, b)$ so that $(b, a) \in E_{k+1}$ in the DAG $D_{k+1}$. Consider the following feasible flow allocation $(f'_{ij})$ on the DAG $D_{k+1}$:
\[
f'_{ij} = \begin{cases}
	\epsilon& \text{if $(i, j)=(b, a)$} \\
	0 = f_{ji} & \text{if  $(i, j) \neq (b, a)$ but $(j, i)$ is reversed}\\
	f_{ij} & \text{if $(i, j)$ is not reversed}
	\end{cases}
\]
where $\epsilon < q_{k,b}^{\text{min}}$ is a sufficiently small value. In other words, the flow allocation $(f'_{ij})$ is formed by reversing links and keeping the previous flow allocation $(f_{ij})$ except that we forward an $\epsilon$-amount of overload traffic from node $b$ to $a$. Let $\widehat{\bm{q}}=(\widehat{q}_{n})_{n\in N}$ be the resulting queue overload vector. We have
\begin{align*}
\widehat{q}_{b} &= q_{k, b}^{\text{min}} - \epsilon < q_{k, b}^{\text{min}}, \text { } \widehat{q}_{a} = \epsilon > q_{k, a}^{\text{min}}=0, \text{ and }\\
\widehat{q}_{n} &= q_{k, n}^{\text{min}}, \ n\notin \{a, b\}.
\end{align*}
Therefore, $\widehat{\bm{q}} <_{\lex} \bm{q}_{k}^{\text{min}}$ (see Fig.~\ref{fig: monotonicity} for an illustration). Let $\bm{q}_{k+1}^{\text{min}}$ be the lexicographically minimal overload vector in $D_{k+1}$. It follows that $\bm{q}_{k+1}^{\text{min}} \leq_{\lex} \widehat{\bm{q}} <_{\lex} \bm{q}_{k}^{\text{min}}$, completing the proof.
\end{proof}



\begin{theorem} \label{thm: main}
Suppose the traffic demand is feasible in $G$, i.e., $\lambda\leq f^{\max}$, and the routing policy induces the overload vector $\bm q^{\text{\emph{min}}}_k$ at every iteration $k$. Then, the link-reversal algorithm will find a DAG whose maximum flow supports $\lambda$ in a finite number of iterations.
\end{theorem}
\begin{proof}[of Theorem~\ref{thm: main}]
The link-reversal algorithm creates a sequence of DAGs $\{D_{0}, D_{1}, D_{2}, \ldots\}$ in which a strict improvement in the lexicographically minimal overload vector is made after each iteration, i.e.,
\[
\bm{q}_{0}^{\text{min}} >_{\lex} \bm{q}_{1}^{\text{min}} >_{\lex} \bm{q}_{2}^{\text{min}} >_{\lex} \cdots.
\]
The lexicographically minimal overload vector is unique in a DAG by Lemma~\ref{lem:201}, the DAGs $\{D_{0}, D_{1}, D_{2}, \ldots\}$ must all be distinct. Since there are a finite number of unique DAGs in the network, the link-reversal algorithm will find a DAG $D_{k^{*}}$ that has the lexicographically minimal overload vector $\bm{q}_{k^{*}}^{\text{min}}= \bm{0}$ and the maximum flow $f_{k^{*}}^{\text{max}} \geq \lambda$ in a finite number of iterations; this DAG $D_{k^{*}}$ exists because the undirected graph $G$ has the maximum flow $f^{\max} \geq \lambda$.
\end{proof}

\subsection{Arrivals outside stability region}
We show that even when $\lambda > f^{\max}$, the link reversal algorithm will stop reversing the links in a finite number of iterations, and it will obtain the DAG that supports the maximum throughput $f^{\max}$. We begin by examining the termination condition of our algorithm and show that if the algorithm stops at iteration $k$, then the DAG $D_k$ supports the max-flow of the network. 

\begin{lemma}
Consider the situation when $\lambda > f_k^{\max}$. If there  is no link $(i,j)$ such that $q_{k,i}^{\min}=0$ and $q_{k,j}^{\min}>0$, then $f^{\max} = f_k^{\max}$ and $\lambda > f^{\max}$. That is, if there are no links to reverse at iteration $k$, and ${\bf q}_k^{\min} > 0$, then the throughput of $D_k$ is equal to $f^{\max}$.
\end{lemma}
\begin{proof}
Let $A_k$ be the set of overloaded nodes under a flow allocation that induces the lexicographically minimal overload vector $\textbf{q}_k^{\min}$ in the DAG $D_k$. We know that $(A_k,A_k^c)$ is a min-cut of the network from Lemma~\ref{MIN-CUT} (in the appendix), so $$cap_k(A_k,A_k^c) =  f_k^{\max}.$$

Suppose the link reversal algorithm stops after iteration $k$, i.e. at iteration $k$ there are no links to reverse. In this situation, there  is no link $(i,j)$ such that $q_{k,i}^{\min}=0$ and $q_{k,j}^{\min}>0$, so by property (9), all the links between $A_k$ and $A_k^c$ go from $A_k$ to $A_k^c$. The capacity of the cut $(A_k,A_k^c)$ is given by 
$$cap_k(A_k,A_k^c) = \sum_{i\in A_k, j \in A_k^c} c_{ij}.$$
This is equal to the capacity of the cut $(A_k, A_k^c)$ in the undirected network $G$. So $f^{\max} \le cap_k(A_k,A_k^c) = f_k^{\max}$. Because $f^k_{\max}$ cannot be greater than $f^{max}$, $f^k_{\max}=f^{\max}$. By assumption $\lambda > f_k^{\max}$, so $\lambda > f^{\max}$.
\end{proof}

When $\lambda>f^{\max}$, this lemma shows that the link reversal algorithm stops only when the DAG achieves the maximum throughput of the network.  Hence, if the DAG doesn't support the maximum throughput, then there exists a link that can be reversed. After each reversal, Lemma \ref{lem: graph is dag} holds, so the directed graph obtained after the reversal is acyclic. We can modify Lemma \ref{lem: monotonicity} to show that every reversal produces a DAG that supports  an improved lexicographically optimal overload vector. We can combine these results to prove the following theorem.
\begin{theorem}
Suppose the traffic demand is not feasible in $G$, i.e., $\lambda > f^{\max}$, and the routing policy induces the overload vector $\bm q^{\text{\emph{min}}}_k$ at every iteration $k$. Then, the link-reversal algorithm will find a DAG whose maximum flow supports $f^{max}$ in a finite number of iterations.
\end{theorem}

\section{Distributed Dynamic Algorithm} \label{dynamic algorithm}

In the previous sections we developed a link reversal algorithm based on the assumption that we had a routing policy that lexicographically minimized the overload vector ${\bf q}^{\min}_k$. The algorithm reversed all the links that went from the set of all the non-overloaded nodes $A_k^c$ to the set of overloaded nodes $A_k$. We showed that repeating this process for some iterations results in a DAG that supports the arrival rate $\lambda$. 

The goal of this paper is to develop a link reversal algorithm based on backpressure. To achieve this goal, we develop a threshold based algorithm that identifies the cut $(A_k, A_k^c)$ using the queue backlog information of backpressure. We can use this cut to perform the same link reversals as in Section \ref{sec:link_reversal} without computing the lexicographically minimum overload vector. Hence, this algorithm generates the same sequence of DAGs as the link reversal algorithm described in the previous section. So, all the previous theorems hold, and the new algorithm will obtain the DAG that supports the arrival rate $\lambda$ (when possible). We will call this algorithm the loop free backpressure (LFBP) algorithm.



We begin by creating an initial DAG $D_0$ using 
 the method presented in Section \ref{sec: initial dag}. Then, we use the backpressure algorithm to route the packets from the source to the destination over $D_0$.
 Let $Q_n(t)$ be the queue length  at node $n$ in slot $t$. The backpressure algorithm can be written as in Algorithm \ref{alg: backpressure}. It simply sends packets on a link $(i,j)$ if node $i$ has more packets than $j$.

\begin{algorithm} [h]
\caption{Backpressure algorithm (BP)}
\label{alg: backpressure}
\begin{algorithmic}[1]
\ForAll{$(i,j) \in E_k$} 
	\If{$Q_i(t) > Q_j(t)$}
		\State Transmit $\min\{c_{ij}, Q_i(t)\}$ packets from  $i$ to $j$
	\EndIf
\EndFor
\end{algorithmic}
\end{algorithm}

Since backpressure is throughput optimal \cite{tassiulas}, if the arrival rate is less than  $ f_0^{\max}$, then all queues are stable.
If the arrival rate is larger than $ f_0^{\max}$, the system is unstable and the queue length  grows at some nodes. In this case, the next lemma shows that if we were using a routing policy that produced the optimal overload vector ${\bf q}^{\min}_k$, the set of all the overloaded nodes $A_k$ and the non-overloaded nodes $A_k^c$  form the smallest min-cut of the DAG $D_k$.

\begin{definition}\label{def:smallest min-cut}
We define the smallest min-cut $(X^*,X^{*c})$ in the DAG $D_{k}$ as the min-cut with the smallest number of nodes in the source side of the cut, i.e., $(X^*,X^{*c})$ solves
\begin{align*}
\text{minimize: } & |X|\\
\text{subject to: } & (X,X^c) \text{ is a min-cut of } D_k.
\end{align*}
\end{definition}

\begin{lemma}\label{SMALLEST MIN-CUT}
Let $A_{k}$  the set of overloaded nodes under a flow allocation $(f_{ij})$ that induces the lexicographically minimum overload vector in the DAG $D_{k}$. If $|A_{k}| > 0$, then $(A_{k}, A_{k}^c)$ is the unique smallest min-cut in $D_{k}$.
\end{lemma}
\begin{proof}[of Lemma~\ref{SMALLEST MIN-CUT}]
The proof is in Appendix~\ref{app:B}.
\end{proof}


Essentially, at every iteration, the link reversal algorithm of Section \ref{sec:link_reversal} discovers the smallest min-cut $(A_k,A_k^c)$ of the DAG $D_k$ and reverses the links that go from $A_k^c$ to $A_k$. Now the following theorem shows that the backpressure algorithm can be augmented with some thresholds to identify the smallest min-cut.

\begin{theorem} \label{BP_THRESHOLD_PROOF}
Assume that $(A_k,A_k^c)$ is the smallest min-cut for DAG $D_k$ with a cut capacity of $f_k^{\max} = cap (A_k,A_k^c) < \lambda$. If packets are routed using the backpressure routing algorithm, then there exist finite constants $T$ and $R$ such that the following happens:
\begin{enumerate}
\item For some $t<T$, $Q_n(t)>R$ for all $n\in A_k$, and
\item For all $t$, $Q_n(t) < R$ for $n\in A_k^c$.
\end{enumerate}
\end{theorem}

\begin{proof}
We will prove the two claims separately. To prove the first claim we will use the fact that the network is overloaded and bottlenecked at the cut $(A_k,A_k^c)$. We will prove the second claim using the fact that the number of packets that arrive into $A_k^c$ in each time-slot is upper-bounded by $f_k^{\max}$, and any cut in the network has a capacity larger than or equal to $f_k^{max}$. The detailed proofs for both claims are given in the Appendix \ref{app:bpThresholdProof}.
\end{proof}

 Each node $n$ has a threshold-based smallest min-cut detection mechanism.
When we start using a particular DAG $D_k$, in each time-slot, we check whether the queue crosses a prespecified threshold $R_k$. Any queue that crosses the threshold gets marked as overloaded. After using the DAG $D_k$ for $T_k$ timeslots, all the nodes that have their queue marked overloaded form the set $A_k$. When the time $T_k$ and threshold $R_k$ are large enough, the cut $(A_k,A_k^c)$ is the smallest min-cut as proven in Theorem \ref{BP_THRESHOLD_PROOF}. After determining the smallest min-cut, an individual node can perform a link reversal by comparing its queue's overload status with its neighbor's. All the links that go from a non-overloaded node to an overloaded node are reversed to obtain  $D_{k+1}$. The complete LFBP algorithm is given in Algorithm \ref{alg: lfbp}.

\begin{algorithm}[h] 
\caption{LFBP (Executed by node $n$)}
\label{alg: lfbp}
\begin{algorithmic}[1]
\State Input: sequences $\{T_k\}, \{R_k\}$, unique ID $n$
\State Generate initial DAG $D_0$ by directing each link $\{n,j\}$ to $(n,j)$ if $n<j$, to $(j,n)$ if $j>n$.
\State Mark the queue $Q_n$ as not overloaded
\State Initialize $t \gets 0$, $k \gets 0$
\While {true}
	\State Use BP to send/recive packets on all links of node $n$
	\If {$(Q_n(t)>R_k)$} \State {Mark $Q_n$ as overloaded.} \EndIf
	\State $t \gets t+1$
\State	
	\State $T_k \gets T_k-1$
	\If {$T_k = 0$} 
		\State Reverse all links $(j,n)$ such that $Q_j$ is not overloaded and $Q_n$ is overloaded.
		\State $k\gets k+1$
		\State Mark $Q_n$ as not overloaded
	\EndIf
	
\EndWhile
\end{algorithmic}
\end{algorithm}

\begin{corollary}
Suppose the traffic demand is feasible in $G$, i.e., $\lambda\leq f^{\max}$. Then, the LFBP algorithm (Algorithm 3) will find a DAG, whose maximum flow supports $\lambda$, in a finite number of iterations.
\end{corollary}
\begin{proof}
Theorem 3 shows that LFBP identifies the smallest min-cut $(A,A^c)$ for the DAG $D_k$. Lemma 6 shows that $A$ is the set of overloaded nodes, and $A^c$ is the set of non-overloaded nodes in a flow allocation that induces the lexicographically minimal overload vector. LFBP reverses the links going from $A^c$ to $A$, which is also the reversals performed by the link reversal algorithm (Algorithm 1). Hence, by Theorem 1, LFBP obtains the DAG that supports $\lambda$.
\end{proof}

Good choices for the thresholds $T_k$ and $R_k$ are topology dependent. When the value of $R_k$ is too small, nodes that are not overloaded might cross the threshold producing a false positive. If the value of $R_k$ is large but $T_k$ is small, the overloaded nodes might not have enough time to develop the backlog to cross $R_k$ which produces false negatives. Hence, a good strategy is to choose a large $R_k$ so that the non-overloaded nodes don't (or rarely) cross this threshold, then chose a large $T_k$ such that the overloaded nodes have enough time to build the backlog to cross $R_k$. Optimizing these thresholds requires further research. Note that our algorithm performance degrades graciously with false positives/negatives. Even when it detects the smallest min-cut incorrectly, the actions of the algorithm preserve the acyclic structure. Thus, in the subsequent iterations the algorithm can improve the DAG again.


\subsection{Algorithm modification for topology changes}
In this section we consider networks with time-varying topologies, where several links of graph $G$ may appear or disappear over time. Although the DAG that supports $\lambda$ depends on the topology of $G$, our proposed policy LFBP can adapt to the topology changes and efficiently track the optimal solution.Additionally, the loop free structure of a DAG is preserved under link removals. Thus, if some of the links in the network disappear, we may continue using LFBP on the new network. 

To handle the appearance of new links in the network smoothly, we will slightly extend LFBP to guarantee the  loop free structure.
 For a DAG $D_k$, every node $n$ stores a unique state $x_n(k)$ representing its position in the topological ordering of the DAG $D_k$. The states are maintained such that they are unique and all the links go from a node with the lower state to a node with the higher state. When a new link $\{i,j\}$ appears we can set its direction to go from $i$ to $j$ if $x_i(k) < x_j(k)$ and from $j$ to $i$ otherwise. Since this assignment of direction to the new link is in alignment with the existing links in the DAG, the loop-free property is preserved.

The state for each node $n$ can be initialized using the unique node ID during the initial DAG creation, i.e. $x_n(0) = n$. Then whenever a reversal is performed the state of node $n$ can be updated as follows: 
\begin{align*}
x_n(k) =\left\{
\begin{array}{ll}
 x_n(k-1) - 2^k\Delta, & \text{ if $n$ is overloaded,}\\
x_n(k-1), & \text{ otherwise.}
\end{array}
\right.
\end{align*}
Here, $\Delta$ is some constant chosen such that $\Delta > \max_{i,j \in N} x_i(0) - x_j(0)$. Note that this assignment of state is consistent with the way the link directions are assigned by the link reversal algorithm. The states for the non-overloaded nodes are unchanged, so the links between these nodes are unaffected. Also, the states for all the overloaded nodes are decreased by the same amount $2^k\Delta$, so the direction of the links between the overloaded nodes is also preserved. Furthermore, the quantity $-2^k\Delta$ is less than the lowest possible state before the $k$th iteration, so the overloaded nodes have a lower state than the non-overloaded nodes. Hence, the links between the overloaded and non-overloaded nodes go from the overloaded nodes to the non-overloaded nodes. 

In this scheme, the states $x_n$ decrease unboundedly as more reversals are preformed. In order to prevent this, after a certain number of reversals, we can rescale the states by dividing them by a large positive number. This decreases the value of the state while maintaining the topological ordering of the DAG. The number of reversals $k$ can be reset to 0, and a new $\Delta$ can be chosen such that it is greater than the largest difference between the rescaled states.

\section{Complexity analysis} \label{sec:complexity}
To understand the number of iteration the link-reversal algorithm takes to obtain the optimal DAG, we analyze the time complexity of the algorithm. 
\begin{theorem}\label{thm:complexity}
Let $C$ be a vector of  the capacities of all the links in $E$, and let $I$ be the set of indices $1,2,..., |E|$.  Define $\delta > 0$ to be the smallest positive difference between the capacity of any two cuts. Specifically, $\delta$ is the solution  of the following optimization problem
\begin{align*}
\min _{A,B\subseteq I} &\sum_{a\in A} c_a-\sum_{b \in B}c_b\\
\text{subject to: } &\sum_{a\in A} c_a>\sum_{b \in B}c_b.
\end{align*}
The number of iterations taken by the link reversal algorithm before it stops is upper bounded by $\lceil|N| \frac{f^{\max}}{\delta} \rceil$ , where $f^{\max}$ is the max-flow of the undirected network.
\end{theorem}
\begin{proof}
After each iteration of the link-reversal algorithm, either the max-flow of the DAG  increases, or the max-flow stays the same and the number of nodes in the source side of the smallest min-cut increases (see Lemma \ref{BIGGER CUT} in the Appendix). We can bound the number of consecutive iterations such that there is no improvement in the max-flow. In particular, every such iteration will add at least one node to the source set. So, it is impossible to have more than $|N|-2$ such iteration. Hence, every $|N|$ iterations we are guaranteed to have at least one increase in the max-flow.

Max-flow is equal to the min-cut capacity, and min-cut capacity is defined as the sum of link capacities. Say, the max-flow of DAG $D_{k+1}$ is greater than that of $D_k$. Let $A$ be the set of indices (in the capacity vector $C$) of the links  in the min-cut of $D_{k+1}$ , and $B$ be the set of indices of the links in the min-cut of $D_{k}$. This choice of A and B forms a feasible solution to the optimization problem given in the theorem statement. Since the optimal solution $\delta$ lower bounds all the feasible solutions in the minimization problem, the increase in the max-flow must be greater than or equal to $\delta$. 

Every $|N|$ iteration the max-flow increases at least by $\delta$. Hence, the DAG supporting the max-flow $f^{\max}$ is formed within $\lceil|N| f^{\max}/\delta \rceil$ iterations.
\end{proof}

\begin{corollary} \label{cor:integer}
In a network where all the link capacities are rational with the least common denominator $\mathcal{D} \in \mathbb{N}$, the number of iterations is upper bounded by $(|N|\mathcal{D}f^{\max})$.
\end{corollary}
\begin{proof}
Since the capacities are rational we can write the capacity of the $i^{\text{th}}$ link as $c_{i}=\frac{\mathcal{N}_i}{\mathcal{D}}$, where $\mathcal{N}_i$ is a natural number. From the definition of $\delta$ in Theorem \ref{thm:complexity}, we get $\delta$ to be the value of the following optimization problem:
\begin{align*}
\min _{A,B\subseteq I} &\frac{1}{\mathcal{D}} \left( \sum_{a\in A} \mathcal{N}_a-\sum_{b \in B}\mathcal{N}_b \right)\\
\text{subject to: } & \sum_{a\in A} \mathcal{N}_a > \sum_{b \in B}\mathcal{N}_b.
\end{align*}
All the $\mathcal{N}_{(.)}$ are integers, so to satisfy the constraint we must have the difference $\sum_{a\in A} \mathcal{N}_a - \sum_{b \in B}\mathcal{N}_b \ge 1$. Hence $\delta \ge \frac{1}{\mathcal{D}}$. Using this value of $\delta$ in Theorem \ref{thm:complexity}, we can see that the number of iterations is upper bounded by $(|N|\mathcal{D}f^{\max})$.
\end{proof}

\begin{corollary}
In a network with unit capacity links, the number of iterations the link-reversal algorithm takes to obtain the optimal DAG is upper bounded by $|N||E|$.
\end{corollary}
\begin{proof}
The max-flow $f^{\max} \le |E|$. So, by Corollary \ref{cor:integer}, the number of iterations is upper bounded by $|N||E|$.
\end{proof}

We conjecture that these upper bounds are not tight, and finding a tighter bound will be pursued in the future research. We simulated the link reversal algorithm in 50,000 different Erdos-Renyi networks ($p=0.5$) of sizes 10 to 50 with randomly assigned link capacities. The link reversal algorithm started with a random initial DAG. We found that it took less than 2 iterations on average to find the optimal DAG.

A worst case lower bound for the number of iteration is $|N|$. This lower bound can be achieved in a line network where the initial DAG has all of its links in the wrong direction.

\section{Simulation Results} \label{simulations}
We compare the delay performance of the LFBP algorithm and the BP algorithm via simulations. 
We will see that the network with the LFBP routing has a smaller backlog on average under the same load. This shows that the LFBP algorithm has a better delay performance. 
We consider two types of networks for the simulations: a simple network with fixed topology, and a network with grid topology where the links appear and disappear randomly. 

\subsection{Fixed topology}
We consider a network with the topology shown in Figure \ref{fig: network_topology}. The edge labels  represent the link capacities. The undirected network has the maximum throughput of 15 packets per time slot. Figure \ref{fig: dag0} shows the initial DAG $D_0$. Instead of running the initial DAG algorithm of Section~\ref{sec: initial dag}, here we choose a zero throughput DAG to test the worst-case performance of LFBP. The arrivals to the network are Poisson with rate $\lambda = 15\rho$ , where we vary $\rho = .5, .55, ..., .95$. For the LFBP algorithm, we set the overload detection threshold to $R_k=60$ for all $n,k$. To choose this parameter, we observed that the backlog buildup in normal operation rarely raises above 60. We also choose the detection period $T_1 = 150$ and $T_k = 50$ for all $k>1$. This provides enough time for buildup, which improve the accuracy of the overload detection mechanism.

We simulate both algorithms for one million slots, using the same arrival process sample path. Figures \ref{fig: dag1} - \ref{fig: dag3} show the various DAGs that are formed by the LFBP algorithm at iterations $k=1,2,3$. We can see that the nodes in the smallest min-cut get overloaded and the link reversals gradually improve the DAG until the throughput optimal DAG is reached.

\begin{figure}[h!]
\centering
\subfigure[Network topology.]{
\begin{overpic}[scale=.42]{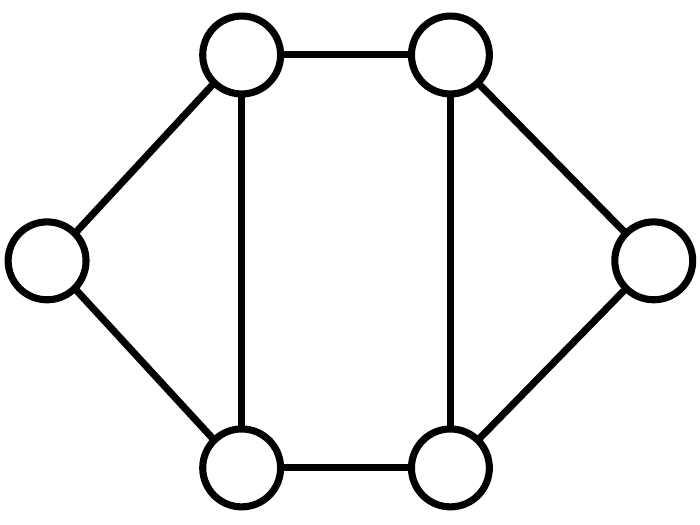}
	\put (4.5,36) {\small $s$}
	\put (32,64) {\small 2}
	\put (62,64) {\small 3}
	\put (91,35) {\small $d$}
	\put (32,5) {\small 1}
	\put (62,5) {\small 4}

	\put (20,44) {\small 15}
	\put (20,26) {\small 5}

	\put (72,44) {\small 15}
	\put (72,26) {\small 5}

	\put (37,35) {\small 5}
	\put (55,35) {\small 5}

	\put (45,59) {\small 5}
	\put (45,10) {\small 10}
\end{overpic}
\label{fig: network_topology}
}
\subfigure[The initial DAG chosen so that LFBP requires several iteration to reach the optimal.]{
\begin{overpic}[scale=.42]{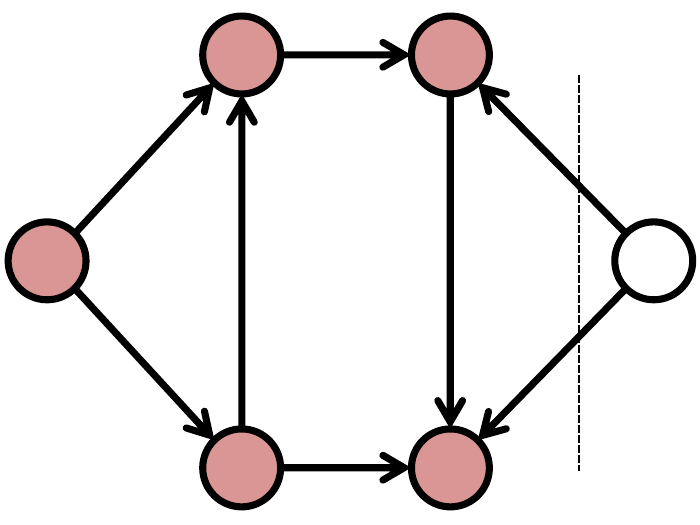}
	\put (4.5,36) {\small $s$}
	\put (32,64) {\small 2}
	\put (62,64) {\small 3}
	\put (91,35) {\small $d$}
	\put (32,5) {\small 1}
	\put (62,5) {\small 4}
\end{overpic}
\label{fig: dag0}
}
\subfigure[After 1st reversal.]{
\begin{overpic}[scale=.42]{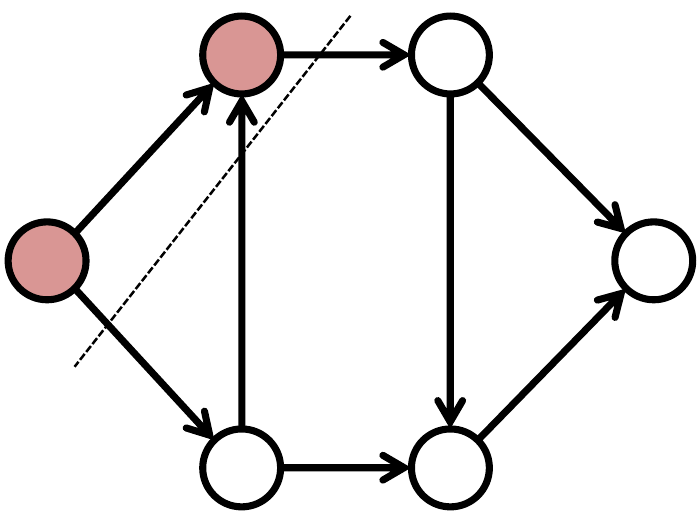}
	\put (4.5,36) {\small $s$}
	\put (32,64) {\small 2}
	\put (62,64) {\small 3}
	\put (91,35) {\small $d$}
	\put (32,5) {\small 1}
	\put (62,5) {\small 4}
\end{overpic}
\label{fig: dag1}
}
\subfigure[After 2nd reversal.]{
\begin{overpic}[scale=.42]{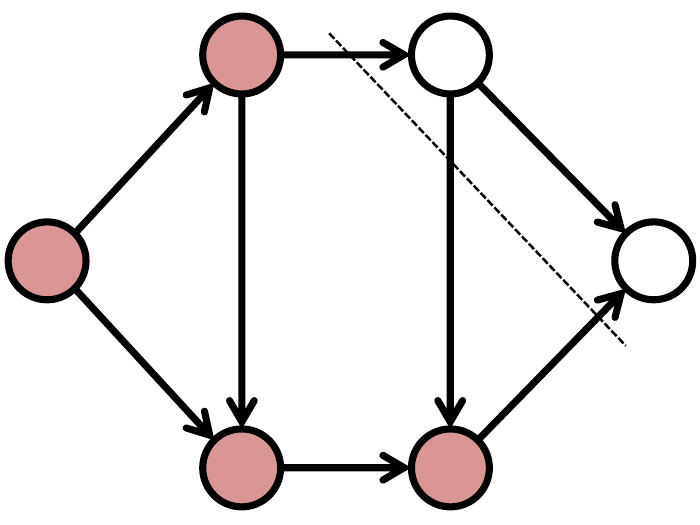}
	\put (4.5,36) {\small $s$}
	\put (32,64) {\small 2}
	\put (62,64) {\small 3}
	\put (91,35) {\small $d$}
	\put (32,5) {\small 1}
	\put (62,5) {\small 4}
\end{overpic}
\label{fig: dag2}
}
\subfigure[The optimal DAG.]{
\begin{overpic}[scale=.42]{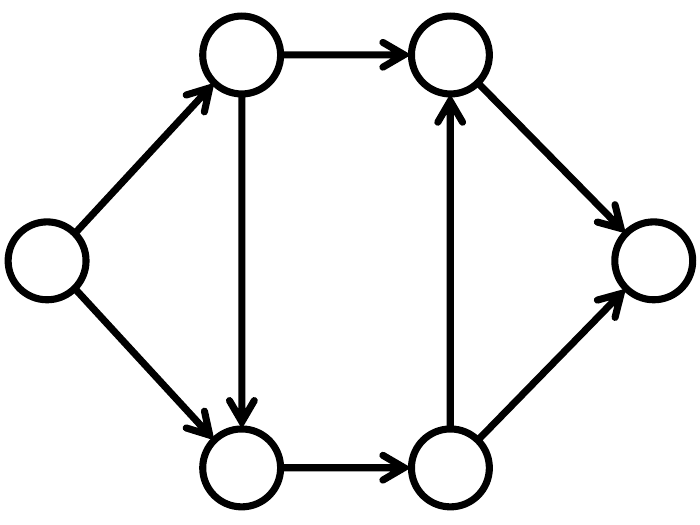}
	\put (4.5,36) {\small $s$}
	\put (32,64) {\small 2}
	\put (62,64) {\small 3}
	\put (91,35) {\small $d$}
	\put (32,5) {\small 1}
	\put (62,5) {\small 4}
\end{overpic}
\label{fig: dag3}
}
\caption{Figure (a) depicts the original network. Figures (b)-(e) are the various stages of the DAG. The red nodes represent the overloaded nodes, and the dashed line shows the boundary of the overloaded and the non-overloaded nodes.}
\end{figure}
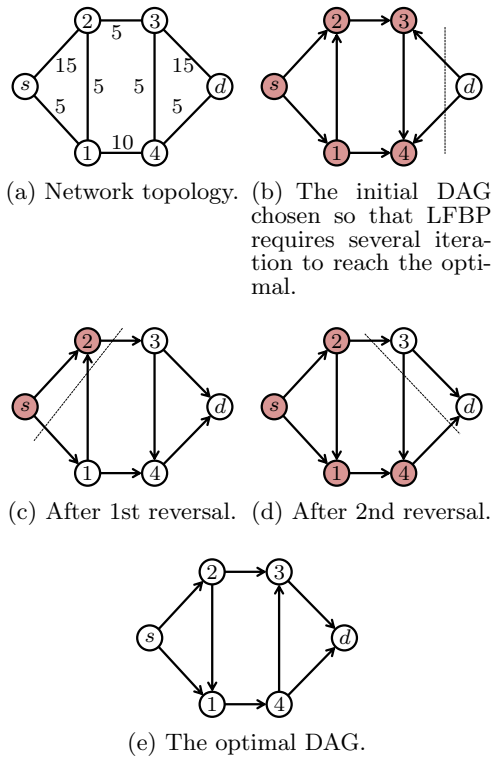

Figure \ref{fig: fixed_topology_delay} compares the total average backlog in the network for  BP  and  LFBP, which is indicative of the average delay.
A significant delay improvement is achieved by LFBP, for example at load $0.5$ the average delay is reduced by $66\%$
 We observe that the gain in the delay performance is more pronounced when the load is low.
 In low load situations, the network doesn't have enough ``pressure'' to drive the packets to the destination and so under BP  the packets go in loops.
 
\begin{figure}[h!]
\centering
\includegraphics[scale=.8]{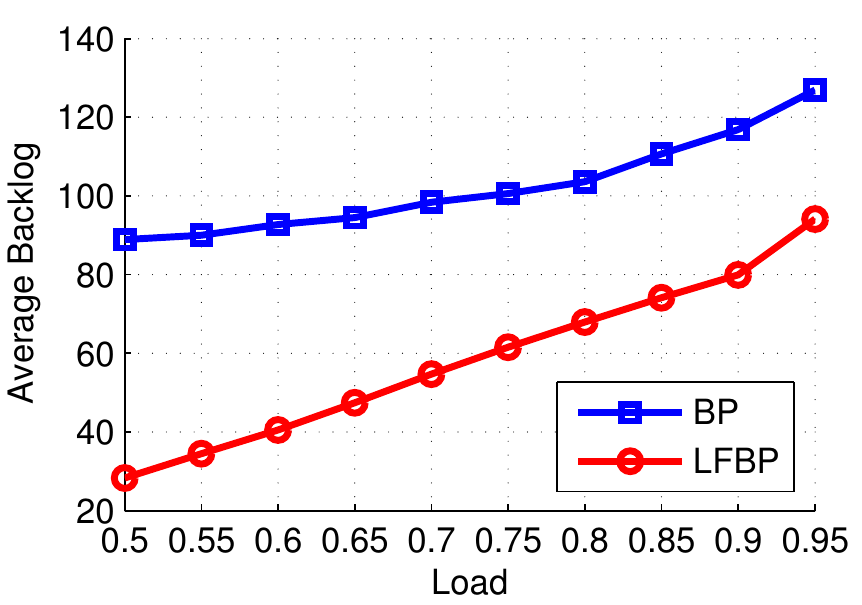}
\caption{Average backlog in the network (Fig. \ref{fig: network_topology})  with fixed topology for the Loop Free Backpressure (LFBP) and the Backpressure (BP) algorithms.} 
\label{fig: fixed_topology_delay}
\end{figure}

\subsection{Randomly changing topology}
To understand the delay performance of the LFBP algorithm on networks with randomly changing topology, we consider a network where 16 nodes are arranged in a $4\times 4$ grid. All the links are taken to be of capacity six. For the LFBP algorithm, we choose a random initial DAG with zero throughput shown in Figure \ref{fig: grid_initial_dag}. The source is on the upper left corner (node 1) and the destination is on the bottom right (node 16).

\begin{figure}[h!]
\centering
\begin{overpic}[scale=.7]{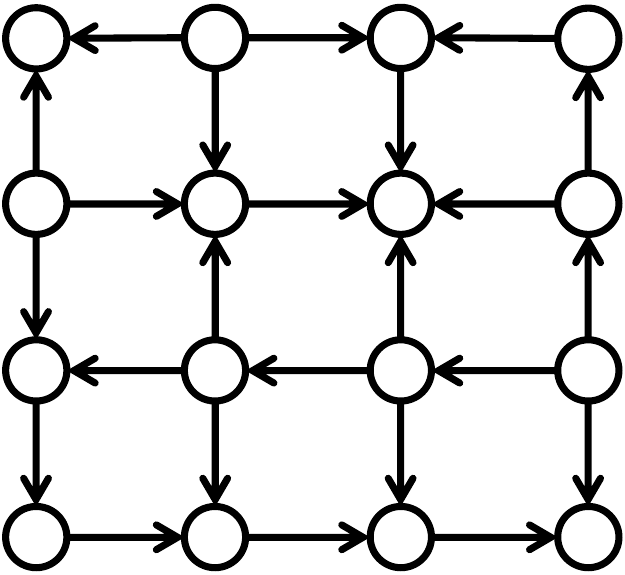}
\put(4,83){\small $1$}
\put(33,83){\small $5$}
\put(62,83){\small $9$}
\put(91,83){\small 13}

\put(4,57){\small $2$}
\put(33,57){\small $6$}
\put(61,57){\small $10$}
\put(91,57){\small 14}

\put(4,30){\small $3$}
\put(32,30){\small $7$}
\put(61,30){\small $11$}
\put(91,30){\small 15}

\put(4,3){\small $4$}
\put(33,3){\small $8$}
\put(61,3){\small $12$}
\put(91,3){\small 16}
\end{overpic}
\caption{Initial DAG for the LFBP algorithm chosen so that the LFBP needs several iterations to reach the optimal DAG. All the links have capacity six.}
\label{fig: grid_initial_dag}
\end{figure}
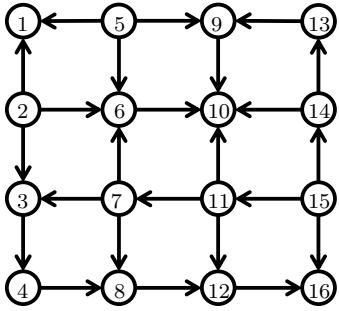

In the beginning of the simulations  all 24 network links are activated. At each time slot an active link fails with a probability $10^{-4}$ and an inactive link is activated with a probability $10^{-3}$. The maximum throughput of the undirected network without any link failures is 12. Clearly on average, each link is ``on'' a fraction $\frac{10}{11}$ of the time, and thus the average maximum throughput of the undirected network with these link failure rates is $\frac{10}{11}\times 12 = 10.9$. 
The arrivals to the networks are Poisson with rate $\lambda = 10.9\rho$, where $\rho = .1, .2, ..., .6$. For the LFBP algorithm, the detection threshold is set to $R_k = 100$ and the detection period is  $T_k = 30$ for all $n,k$. These parameters were chosen so that there are several reversals before a topology change occurs in the undirected network. The simulation was carried out for a million slots. 

Figure \ref{fig: changing_topology_delay} compares the average backlog of LFBP and BP. In the low load scenarios  LFBP reduces delay significantly (by 85\% for load = 0.1) even though the topology changes challenge the convergence of the link-reversal algorithm. As the load increases, both the algorithms begin to obtain a similar delay performance.

\begin{figure}[h]
\centering
\includegraphics[scale=.75]{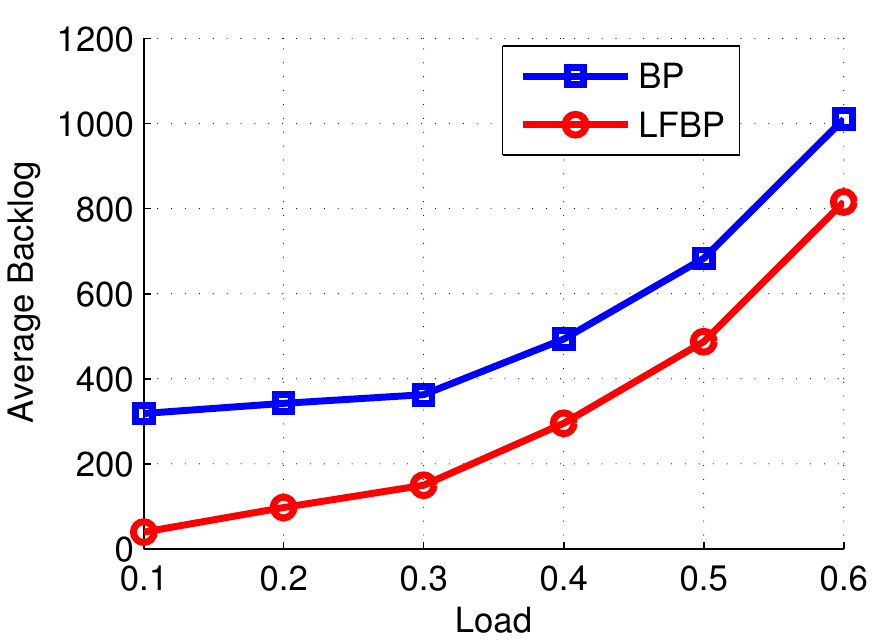}
\caption{Average backlog in the network with random link failures (Fig. \ref{fig: grid_initial_dag}) for the Loop Free Backpressure algorithm and the Backpressure algorithm.}
\label{fig: changing_topology_delay}
\end{figure}

\section{Multicommodity simulation}
We extend of the link reversal algorithm to the networks with multiple commodities. The multi-commodity algorithm is identical to the single commodity algorithm, with the exception that we now use the multicommodity backpressure of \cite{tassiulas}. Each node $n$ maintains a queue $Q^y_n(t)$ for each commodity $y$. Each commodity is assigned its own initial DAG. A pseudocode for the multicommodity LFBP that we used is given in Algorithm \ref{alg:mlfbp}. An important direction for future research is to determine whether the claims proven for a single commodity in the previous sections extend to the multicommodity case.

\begin{algorithm}[h] 
\caption{Multicommodity LFBP (Executed by $n$)}
\label{alg:mlfbp}
\begin{algorithmic}[1]
\State Input: sequences $\{T_k\}, \{R_k\}$, unique ID $n$
\State For each commodity $y$, generate initial DAG $D_0^y$ by directing $\{n,j\}$ to $(n,j)$ if $n<j$, to $(j,n)$ if $j>n$.
\State Mark all queues $Q_n^y$ as not overloaded
\State Initialize $t \gets 0$, $k \gets 0$
\While {true}
	\State Use Multicommodity BP to send/recive packets on all links of node $n$
	\ForAll {$y$}	
		\If {$(Q_n^y(t)>R_k)$} \State {Mark this $Q_n^y$ as overloaded.} \EndIf
	\EndFor
	\State $t \gets t+1$
\State	
	\State $T_k \gets T_k-1$
	\If {$T_k = 0$} 
		\ForAll {$y$}
			\State Reverse links $(j,n)$ in $D_k^y$ if $Q_j^y$ is not overloaded and $Q_n^y$ is overloaded.
		\EndFor
		\State $k\gets k+1$
		\State Mark all queues as not overloaded
	\EndIf
	
\EndWhile
\end{algorithmic}
\end{algorithm}

For the simulation, we consider a network arranged in a $4\times 4$ grid as shown in Figure \ref{fig: grid_initial_dag}. Each link has a capacity of 6 packets per time-slot. There are three commodities in the network defined by the source destination pairs (1,16), (4,13) and (5,8). For the LFBP algorithm, each commodity starts with the same initial DAG given in Figure \ref{fig: grid_initial_dag}. 

We use the arrival rate vector $\lambda^{\max} = [7.18,6.96,9.86]$, which is a max-flow vector for this network computed by solving a linear program. We scale this vector by various load factors $\rho$ ranging from 0.1 to 0.9. The arrivals for each commodity $i$ is Poisson with rate $\rho \lambda^{\max}_i$. In the beginning of the LFBP simulation, $\lfloor 500/\rho \rfloor$ dummy packets are added to the source of each commodity. This is helpful in low load cases because it forces the algorithm to find a DAG with high throughput, and avoids stopping at a DAG that only supports the given (low) load. $R_k$ was chosen to be $50$ and $T_k = 50$ for all $k>0$. The simulation was executed for 500,000 time-steps.

Figure \ref{3commodity_plot} shows the average backlog in the network for different loads under backpressure and multicommodity LFBP. We can see that the LFBP algorithm has a significantly improved delay performance compared to backpressure.

\begin{figure}[h!]
\centering
\includegraphics[scale=.7]{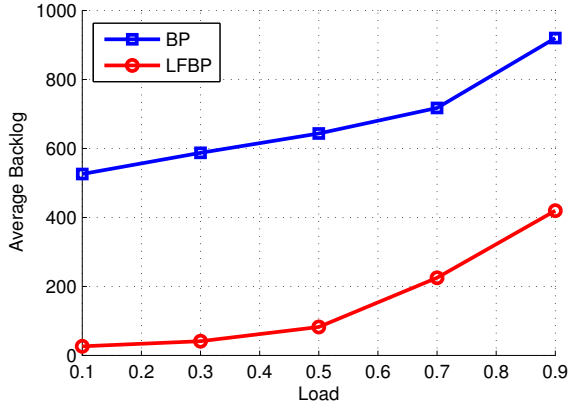}
\caption{Average backlog in a multicommodity network with fixed topology for LFBP and BP algorithms.}
\label{3commodity_plot}
\end{figure}

\section{Conclusion}
Backpressure routing and link reversal algorithms have been separately proposed for mobile wireless networks applications. In this paper we show that these two distributed schemes can be successfully combined to yield good throughput and delay performance.We develop the Loop-Free Backpressure Algorithm which jointly routes packets in a constrained DAG and reverses the links of the DAG to improve its throughput. We show that the algorithm ultimately results in  a DAG that yields the maximum throughput. Additionally, by restricting the routing to this DAG we eliminate loops, thus reducing the average delay. 
Future investigations involve optimization of the overload detection parameters  and studying the performance of the scheme on the networks with  multiple commodities.

\appendix

\section{Lemma \ref{MIN-CUT}}\label{app:A}
\begin{lemma}\label{MIN-CUT}
Consider a DAG $D_k$ with source node $s$, destination node $d$, and arrival rate $\lambda$. Let $A_{k}$ be the set of overloaded nodes under the flow allocation $(f_{ij})$ that yields the lexicographically minimum overload vector. If $|A_{k}| > 0$, then $(A_{k}, A_{k}^c)$ is a min-cut of the DAG $D_{k}$.
\end{lemma}
\begin{proof}[Proof of Lemma~\ref{MIN-CUT}]
First we show that $(A_{k}, A_{k}^c)$ is a cut, i.e., the source node $s\in A_{k}$  and the destination node $d\in A_{k}^{c}$. The destination node $d$ has zero queue overload rate $q_{d}=0$ because it does not buffer packets; hence $d\in A_{k}^{c}$. We show $s \in A_{k}$ by contradiction. Assume $s \notin A_{k}$.  The property~\eqref{property1} shows that there is no flow going from $A_{k}^{c}$ to $A_{k}$, i.e., 
\[
\sum_{(i, j) \in E_{k}:\, i\in A_{k}^{c},\, j\in A_{k}} f_{ij} = 0.
\]
The flow conservation equation applied to the collection $A_{k}$ of nodes yields
\begin{align*}
\sum_{n \in A_{k}} q_n &= \sum_{(i,n)\in E_{k}:\, i\in A_{k}^c,\, n\in A_{k}} f_{in} - \sum_{(n,j)\in E_{k}:\, n \in A_{k},\, j \in A_{k}^c} f_{nj}\\
&=  - \sum_{(n, j)\in E_{k}:\, n \in A_{k}, j \in A_{k}^c} f_{nj} \leq 0,
\end{align*}
which contradicts the assumption that the network is overloaded (i.e., $|A_{k}|>0$). Note that in the above equation $\lambda$ does not appear because of the premise $s\notin A_k$.

By the max-flow min-cut theorem, it remains to show that the capacity of the cut $(A_{k}, A_{k}^{c})$ is equal to the maximum flow $f_{k}^{\text{max}}$ of the DAG $D_{k}$. Under the flow allocation $(f_{ij})$ that induces the lexicographically minimal overload vector, the throughput of the destination node $d$ is the maximum flow $f_{k}^{\text{max}}$ (see Lemma~\ref{lem:201}). It follows that
\begin{align} 
f_{k}^{\text{max}} &= \lambda - \sum_{i\in N} q_{i} = \lambda- \sum_{i\in A_{k}} q_i \label{eq:21011} \\
&= \sum_{(i,j)\in E_{k}:\, i\in A_{k},\, j\in A_{k}^c} f_{ij} \label{eq:212} \\
&= \sum_{(i,j)\in E_{k}:\, i\in A_{k}, j\in A_{k}^c} c_{ij} = \capp_k(A_{k}, A_{k}^{c}). \label{eq:213} 
\end{align}
where~\eqref{eq:21011} uses \eqref{eq:201} and $q_{i}=0$ for all nodes $i\notin A_{k}$,~\eqref{eq:212} follows the flow conservation law over the node set $A_{k}$, and~\eqref{eq:213} uses the property~\eqref{property2} in Lemma~\ref{lem:201}.
\end{proof}

\section{Proof of Lemma \ref{SMALLEST MIN-CUT}}\label{app:B}
\begin{proof}[of Lemma~\ref{SMALLEST MIN-CUT}]
Lemma \ref{MIN-CUT} shows that $(A_{k},A_{k}^c)$ is a min cut of the DAG $D_{k}$. It suffices to prove that if there exists another min-cut $(B, B^c)$, i.e., $A_{k}\neq B$ and $\capp_k(A_{k}, A_{k}^c)=\capp_k(B, B^c)$, then $A_{k} \subset B$. The proof is by contradiction. Let us assume that there exists another min-cut $(B,B^c)$ such that $A_{k} \not\subset B$. We have the source node $s\in A_{k}\cap B$ and the destination node $d\in A_{k}^{c} \cap B^{c}$. Consider the partition $\{C, D, E, F\}$ of the network nodes such that $C=A_{k} \cap B$, $D = A_{k}\backslash B$, $E = B\backslash A_{k}$ and $F = N \backslash (A_{k} \cup B)$ (see Fig.~\ref{partitions}). Since $A_{k}\not\subset B$ and $A_{k}\neq B$, we have $|D| > 0$. Also, we have $s\in C$ and $d \in F$.
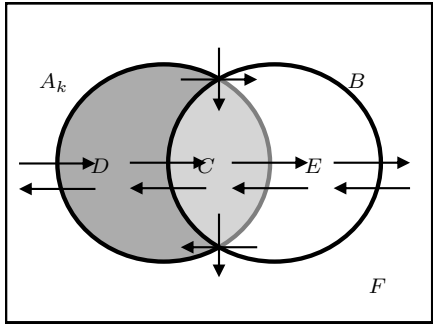
\begin{figure}[ht]
\centering
\begin{overpic}[scale=.8]{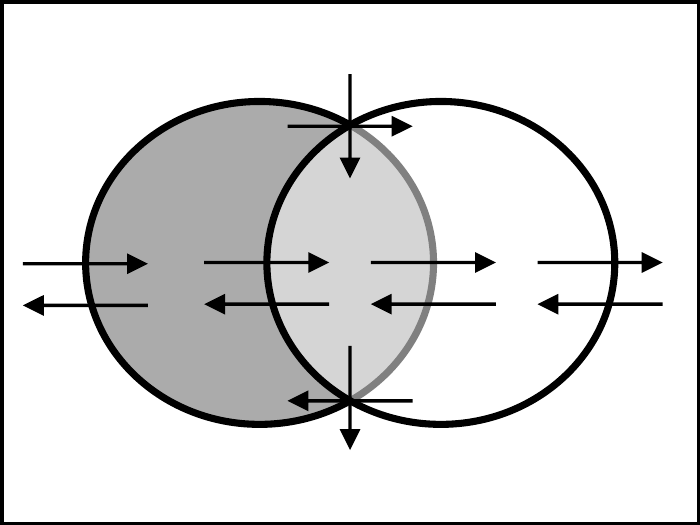}
	\put (8,55) {\small $A_k$}
	\put (80,55) {\small $B$}
	\put (20,35) {\small $D$}
	\put (45,35) {\small $C$}
	\put (70,35) {\small $E$}
	\put (85,7) {\small $F$}
\end{overpic}
\caption{A partition of the node set $N$ where $A_{k} = C\cup D$ and $B = C\cup E$.}
\label{partitions}
\end{figure}
Let $(f_{ij})$ be a flow allocation that yields the lexicographically minimum overload vector in $D_{k}$. Properties~\eqref{property1} and~\eqref{property2} show that 
\begin{align}
f_{ij} &= c_{ij}, \ \forall \, i \in A_{k},\, j \in A_{k}^c, \label{eq:204} \\
f_{ij} &= 0,  \ \forall \, i \in A_{k}^c,\, j\in A_{k}. \label{eq:205}
\end{align}
The capacity of the cut $(B, B^c)$ in the DAG $D_{k}$, defined in~\eqref{eq:202}, satisfies
\begin{equation} \label{eq:206}
\capp_k(B, B^c) = \capp_k(B, D) + \capp_k(B, F),
\end{equation}
where $B^{c}=D\cup F$. Under the flow allocation $(f_{ij})$, we have
\begin{equation} \label{eq:207}
\capp_k(B, D) = \sum_{(i, j)\in E_{k}: i\in B, j\in D} c_{ij} \geq \sum_{(i, j)\in E_{k}: i\in B, j\in D} f_{ij}.
\end{equation}
Applying the flow conservation equation to the collection of nodes in $D$ yields
\begin{equation} \label{eq:203}
\sum_{(i, j)\in E_{k}: i\in B, j\in D} f_{ij} \geq \sum_{i\in D} q_{i} + \sum_{(i, j)\in E_{k}: i\in D, j\in F} f_{ij}.
\end{equation}
In~\eqref{eq:203}, the first term is the sum of incoming flows into the set $D$; notice that there is no incoming flow from $F$ to $D$ because of the flow property~\eqref{eq:205}. The second term is the sum of queue overload rates in $D$. The last term is a partial sum of outgoing flows leaving the set $D$, not counting flows from $D$ to $B$; hence the inequality~\eqref{eq:203}. From the flow property~\eqref{eq:204}, the outgoing flows from the set $D$ to $F$ satisfy
\begin{equation} \label{eq:208}
\sum_{(i, j)\in E_{k}: i\in D, j\in F} f_{ij} = \sum_{(i, j)\in E_{k}: i\in D, j\in F} c_{ij}.
\end{equation}
Combining~\eqref{eq:206}-\eqref{eq:208} yields
\begin{align}
\capp_k(B, B^c) &= \capp_k(B, D) + \capp_k(B, F) \notag \\
&\geq \sum_{i\in D} q_{i} + \sum_{(i, j)\in E_{k}: i\in D, j\in F} c_{ij} + \capp_k(B, F) \notag \\
&> \sum_{(i, j)\in E_{k}: i\in D, j\in F} c_{ij} + \capp_k(B, F) \notag \\
&= \capp_k(A_{k}\cup B, F), \label{eq:209}
\end{align}
where the second inequality follows that all nodes in $D$ are overloaded and $q_{n}>0$ for all $n\in D$. Inequality~\eqref{eq:209} shows that there exists a cut $(A_k\cup B, F)$ that has a smaller capacity, contradicting that $(B, B^{c})$ is a min-cut in the DAG $D_{k}$. Finally, we note that the 
partition
 $(A_k,A^c_k)$ is unique because the lexicographically minimal overload vector is unique by Lemma~\ref{lem:201}.
\end{proof}

\section{Proof of Theorem \ref{BP_THRESHOLD_PROOF}}\label{app:bpThresholdProof}
\begin{proof}[of the first claim]
First we will show that the queue at the source $Q_s(t)$ crosses any arbitrary threshold $R_1$. We know that for some node $n\in A_k$, $Q_n(t) \rightarrow \infty$ as $t \rightarrow \infty$ because the external arrival rate to the source $s\in A_k$ is larger than the rate of departure from set $A_k$, i.e. $\lambda > cap(A_k,A_k^c)$. The backpressure algorithm sends packets on a link (i,j) only if $Q_i(t) > Q_j(t)$. Hence, at any time-slot if a node $b \ne s$ has a large backlog, then one of its parents $p$ must also have a large backlog. $Q_p$ can be slightly smaller than $Q_b$ because $Q_b$ might also receive packets from other nodes at the same time-slot. Specifically, $Q_p(t) > Q_b(t+1)-\sum_i c_{ib}$. Performing the induction on the parent of $p$ we can see that the source node must have a high backlog when any node in $A_k$ develops a high backlog. Note that the network is a DAG and the node $n$ received packets form the source to develop its backlog, so the induction much reach the source node. Hence, when $Q_b(T_1) \gg R_1$, $Q_s(t) > R_1$ for some $t<T_1$.

%

Now we will show that every node in $A_k$ crosses the threshold $R$. Let $B_1 \subseteq A_k$ be the set of nodes such that $Q_n(t)>R_1$ for some time $t<T_1$. We showed that $s\in B_1$. We will show that when $B_1 \ne A_k$, there exists some set $B_2$, such that  (i) $B_1 \subset B_2$, and (ii) for every node $n \in B_2$, $Q_n(t) > R_2$ for some $t < T_2$. Here, $R_2$ and $T_2$ are large thresholds.

Assume $B_1 \ne A_k$. Let $C_1=A_k\backslash B_1$, i.e all nodes in $C_1$ haven't crossed the threshold $R_1$ until time $T_1$. 
Let $c_{B_1C_1}$ be the total capacity of the links going from $B_1$ to $C_1$, and $c_{C_1A_k^c}$ be the total capacity of the links going from $C_1$ to $A_k^c$. We have $c_{B_1C_1}>c_{C_1A_k^c}$ because $(A_k,A_k^c)$ is the smallest min-cut (see Figure \ref{fig:OverloadedSubset}). When the backlogs of the nodes of $B_1$ are much larger than the nodes of $C_1$, the nodes in $C_1$ receive packets from $B_1$ at the rate of $c_{B_1C_1}$ packets per time-slot, and no packets are sent in the reversed direction. The rate of packets leaving the nodes in $C$ is upper bounded by $c_{B_1A_k^c}$ which is smaller than the incoming rate. Hence, at least one node $n' \in C$ must collect a large backlog, say larger than $R_2<R_1$. So, each node in the set $B_2 = B_1 \cup \{n'\}$ have a backlog larger than $R_2$ at some finite time $T_2$.

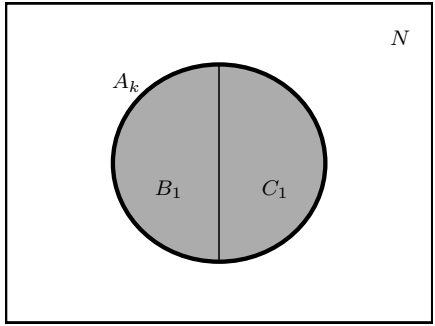
\begin{figure}[ht]
\centering
\begin{overpic}[scale=.8]{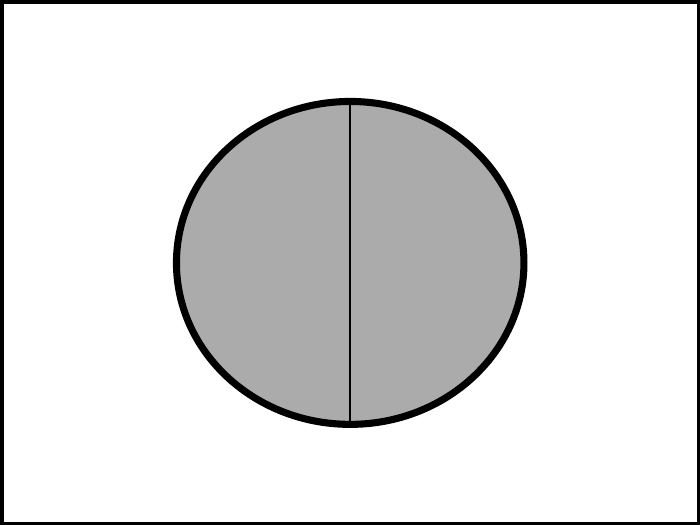}
	\put (25,55) {\small $A_k$}
	\put (35,30) {\small $B_1$}
	\put (60,30) {\small $C_1$}
	\put (90,65) {\small $N$}
\end{overpic}
\caption {Let $(A_k,A_k^c)$ be the smallest min-cut. We showed that $s \in B_1$. Say, $c_{C_1A^c} \ge c_{B_1C_1}$ then the cut $(B_1,B_1^c)$ has the capacity of $c_{B_1A^c} + c_{B_1C_1} \le cap(A_k,A_k^c)$. This contradicts the assumption that $(A_k,A_k^c)$ is the smallest min-cut. So, $c_{C_1A_k^c} < c_{B_1C_1}$. }
\label{fig:OverloadedSubset}
\end{figure}

Now using induction we can see that for $B_m$ where $m < |A_k|$, $B_m = A_k$ and all the nodes in $B_m$ cross a threshold $R = \min\{R_1, ..., R_m\}$ by time $T=\max\{T_1, ..., T_m\}$.
\end{proof}

\begin{proof}[of the second claim]
We will use the following fact to prove this claim: for any subset of nodes $S$, if the number of packets entering $S$ is lower than or equal to the number of packets leaving $S$ on every time-slot, then the total backlog in $S$ doesn't grow. So, the backlog in each node of S is bounded.

Assume a node $b$ develops a backlog $Q_b(t)> R_1$. Here $R_1$ is a chosen such that $$R_1 = |A_k^c| \sum_{i,j \in A_k^c} c_{ij} + \max_{n \in A_k^c} Q_n(0).$$
Consider a subset $B$ of $A_k^c$ such that for every node $i\in B$ and $j \in C = A_k^c \backslash B$, $(Q_i(t) - Q_j(t)) > c_{ij}$. The sets $B$ and $C$ must be nonempty because $Q_b(t)$ is large and $Q_d(t)$ is zero, that is $b \in B$ and $d \in C$. Note that backpressure doesn't send any data from $C$ to $B$. 

\begin{figure}[ht]
\centering
\begin{overpic}[scale=.8]{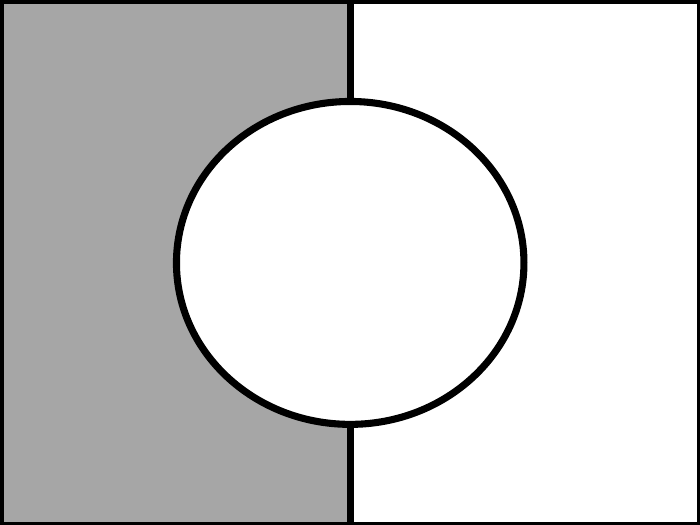}
	\put (46,35) {\small $A_k$}
	\put (5,30) {\small $B$}
	\put (90,30) {\small $C$}
\end{overpic}
\caption {Let $(A_k,A_k^c)$ be the smallest min-cut. We showed that $d \in C$. Say, $c_{AB} > c_{BC}$  then the cut $(B \cup A_k, (B \cup A_k)^c)$ has the capacity of $c_{BC} + c_{A_kC} < c_{AB} + c_{A_k C} =  cap(A_k,A_k^c)$. This contradicts the assumption that $(A_k,A_k^c)$ is the smallest min-cut. So, $c_{AB} < c_{BC}$. }
\label{fig:NonOverloadedSet}
\end{figure}

Let $c_{AB}$ be the capacity of the links going from $A$ to $B$, and let $c_{BC}$ be the capacity of the links going from $B$ to $C$. So, the number of packets entering $B$ at timeslot $t$ is upper bounded by $c_{AB}$. The number of packets leaving $B$ is equal to $c_{BC}$. Since $(A,A^c)$ is the smallest min-cut, $c_{AB} \le c_{BC}$ (see Figure \ref{fig:NonOverloadedSet}). Hence, the number of packets entering $B$ is less than or equal to the number of packets leaving it at time $t$. 

Therefore as soon as one of the nodes crosses threshold $R_1$, the sum backlog becomes bounded. We can choose a threshold $R \gg R_1$ such that this threshold is never crossed by any nodes in $A_k^c$.

\end{proof}

\section{Lemma \ref{BIGGER CUT}}
\begin{lemma}\label{BIGGER CUT}
Consider the case when $\lambda > f_k^{\max}$. The link reversal algorithm is applied on DAG $D_k$ to obtain $D_{k+1}$. Let $(A_k,A_k^c)$ and $(A_{k+1},A_{k+1}^c)$ be the smallest min-cuts of $D_k$ and $D_{k+1}$ respectively. Then, either $cap_k(A_k,A_k^c) > cap_{k+1}(A_{k+1},A_{k+1}^c)$, or $cap_k(A_k,A_k^c) = cap_{k+1}(A_{k+1},A_{k+1}^c)$ and $|A_{k+1}| > |A_k| $
\end{lemma}
\begin{proof}
\begin{figure}[ht]
\centering
\begin{overpic}[scale=.8]{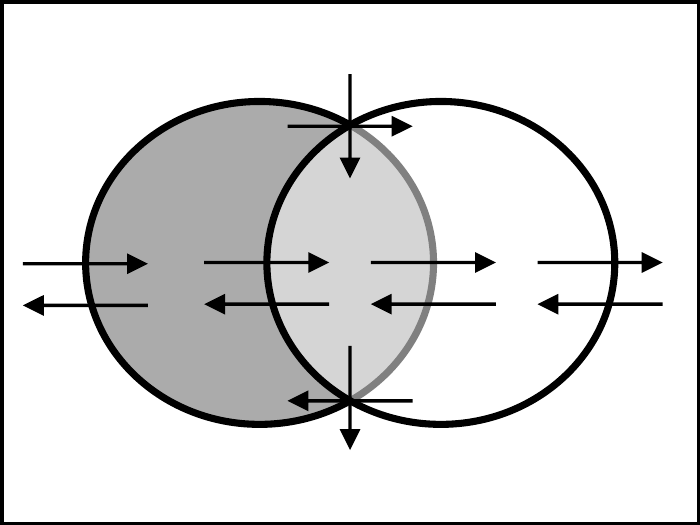}
	\put (8,55) {\small $A_k$}
	\put (82,55) {\small $A_{k+1}$}
	\put (1,39){\small $l_1,l_1'$}
	\put (1,25){\small $l_2,l_2'$}
	\put (27,39){\small $l_3,l_3'$}
	\put (27,25){\small $l_4,l_4'$}
	\put (63,39){\small $l_5,l_5'$}
	\put (63,25){\small $l_6,l_6'$}
	\put (88,39){\small $l_7,l_7'$}
	\put (88,25){\small $l_8,l_8'$}
	\put (45,65){\small $l_9,l_9'$}
	\put (45,5){\small $l_{10},l_{10}'$}
	\put (60,18){\small $l_{11},l_{11}'$}
	\put (25,53){\small $l_{12},l_{12}'$}
\end{overpic}
\caption{Here $l_i$ represents the sum of the capacities of the links going from one partition to the next in the DAG $D_k$, and $l_i'$ represents the sum of the link capacities in the DAG $D_{k+1}$. For example, $l9$ and $l9'$ represent the links that go from $(A_k \cup A_{k+1})^c$ to $(A_k \cap A_{k+1})$ in DAGs $D_k$ and $D_{k+1}$ respectively.}
\label{partitions2}
\end{figure}
Consider the partitioning of the nodes as shown in Figure \ref{partitions2}. For $i=1,...,12$, $l_i$ represents the sum of the capacities of the links going from one partition to the next in the DAG $D_k$, and $l_i'$ represents the sum of the link capacities in the DAG $D_{k+1}$. The capacities of the smallest min-cut, before and after the reversal are given by $$cap_k(A_k,A_k^c) = l_2+l_5+l_{10}+l_{12}\text{ and }$$ $$cap_{k+1}(A_{k+1},A_{k+1}^c) =  l_4' + l_7' + l_{10}' + l_{11}'$$ respectively. Note that only the links that are coming into $A_k$ are different in $D_k$ and $D_{k+1}$. So 
\begin{equation}
l_i=l_i' \text{ for } i=3,4,7,8,10,12. \label{eq1}
\end{equation}
Because of the reversal there are no links coming into $A_k$ in the DAG $D_{k+1}$:
\begin{equation}
l_1',l_6',l_9', l_{11}' = 0 \label{eq2}.
\end{equation}
After the reversal, the incoming links to $A_k$ become outgoing from $A_k$,
\begin{equation}
l_{10}' = l_{10} + l_9 \label{eq3}.
\end{equation}
(Corresponding equations for $l_2', l_5'$ and $l_{12}'$ are omitted because they are not necessary for the proof). Since $(A_k, A_k^c)$ is a min-cut, 
\begin{equation}
l_5 \le l_7. \label{eq4}
\end{equation}
This is true because otherwise the cut $(A_k \cup A_{k+1}, (A_k \cup A_{k+1})^c)$  in the DAG $D_k$ has a smaller capacity then the min cut $(A_k,A_k)^c$. Specifically, let us assume $l_5>l_7$. Then, we get the contradiction: 
\begin{align*}
cap_k(A_k \cup A_{k+1}, (A_k \cup A_{k+1})^c) &= l_2 + l_7 + l_{10}\\
	& < l_2 +l_5 + l_{10} + l_{12} \\
	& = cap_k(A_k, A_k)^c
\end{align*}

First we will show that if $A_k\backslash A_{k+1}\ne \phi$, then the capacity of the DAG must have increased. The proof is by contradiction. 

Let us assume that the throughput didn't increase. So, 
\begin{align}
cap_k(A_k,A_k^c) &\ge cap_{k+1}(A_{k+1},A_{k+1}^c)  \nonumber\\
&=  l_4' + l_7' + l_{10}' + l_{11}'  \nonumber \\
&= l_4 + l_7 + l_{10} + 0  \label{eq5}\\
&\ge l_4 + l_5 + l_{10} \label{eq6}\\
&= cap_k(A_k \cap A_{k+1}, (A_k \cap A_{k+1})^c).
\end{align}
(\ref{eq5}) is follows from (\ref{eq1}) and (\ref{eq2}), and (\ref{eq6}) follows from (\ref{eq4}). Since $A_k\backslash A_{k+1}\ne \phi$ by assumption, $|A_k| > |A_k \cap A_{k+1}|$. This leads to a contradiction, because in DAG $D_k$ the cut $(A_k \cap A_{k+1}, (A_k \cap A_{k+1})^c)$ is smaller than the smallest min-cut $(A_k,A_k^c)$. Hence, $cap_k(A_k,A_k^c) < cap_{k+1}(A_{k+1},A_{k+1}^c)$.

Next, we will consider the case $A_k\backslash A_{k+1} = \phi$. Using (\ref{eq4}),
$$cap_k(A_k,A_k^c) =  l_5+l_{10}\le l_7+l_{10}.$$ In this situation, we again have two cases.
First, if $A_k = A_{k+1}$ we know that $l_{10}' > l_{10}$ and $l_7 = 0$. Hence, $cap_k(A_k,A_k^c) < l_{10}' = cap_{k+1}(A_{k+1}, A_{k+1}^c)$. 

Second, if $A_k \subset A_{k+1}$, then $|A_k| > |A_{k+1}|$ and  
\begin{equation}
l_{10}' \ge l_{10}. \label{eq8}
\end{equation}
Using (\ref{eq1}) and (\ref{eq8}) $cap_k(A_k,A_k^c) \le cap_{k+1}(A_{k+1}, A_{k+1}^c)$.
\end{proof}

\end{document}